\documentclass[11pt]{llncs}

\usepackage{amsmath,amsfonts,amssymb,enumerate}
\usepackage[T1]{fontenc}
\usepackage{hyperref}
\usepackage[noadjust,nospace]{cite}

\usepackage{boxedminipage}

\makeatletter\renewcommand\subparagraph{\@startsection{subparagraph}{5}{0ex} {3.25ex \@plus1ex \@minus .2ex} {-1em} {\normalfont\normalsize\bfseries}}\makeatother

\newcommand{\opts}{opt_S}
\newcommand{\vd}{{\sf vd}}
\newcommand{\ed}{{\sf ed}}
\newcommand{\ea}{{\sf ea}}
\newcommand{\cdpe}{{\sc CDPE}}
\newcommand{\dpe}{{\sc DPE}}
\newcommand{\cdbe}{{\sc CDBE}}
\newcommand{\dbe}{{\sc DBE}}

\newcommand\displaycase[1]{{\bf #1}}

\renewcommand{\P}{{\sf P}}
\newcommand{\NP}{{\sf NP}}
\newcommand{\FPT}{{\sf FPT}}
\newcommand{\W}{{\sf W[1]}}

\title{Editing to Eulerian Graphs\thanks{The research 
leading to these results has received funding from the European Research Council under the European Union's Seventh Framework Programme (FP/2007-2013)/ERC Grant Agreement n. 267959 and from EPSRC Grant EP/K025090/1.
An extended abstract of this paper will appear in the proceedings of FSTTCS 2014.
}} 
\author{Konrad K. Dabrowski\inst{1} \and Petr A. Golovach\inst{2} \and Pim van 't Hof\inst{2} \and\\ Dani{\"e}l Paulusma\inst{1}}

\institute{
School of Engineering and  Computing Sciences, Durham University,\\ 
Science Laboratories, South Road, Durham DH1 3LE, United Kingdom\\
\email{\{konrad.dabrowski,daniel.paulusma\}@durham.ac.uk}
\and
Department of Informatics, University of Bergen,\\
PB 7803, 5020 Bergen, Norway\\ 
\email{\{petr.golovach,pim.vanthof\}@ii.uib.no}
}

\begin{document}

\maketitle

\begin{abstract}
We investigate the problem of modifying a graph into a connected graph in which
the degree of each vertex satisfies a prescribed parity constraint. Let $\ea$,
$\ed$ and $\vd$ denote the operations 
edge addition, edge deletion and vertex
deletion respectively. For any $S\subseteq \{\ea,\ed,\vd\}$, we define
\textsc{Connected Degree Parity Editing$(S)$} (\cdpe($S$)) to be the problem
that takes as input a graph~$G$, an integer $k$ and a function $\delta\colon
V(G)\rightarrow\{0,1\}$, and asks whether $G$ can be modified into a connected
graph $H$ with $d_{H}(v)\equiv\delta(v)~(\bmod~2)$ for each $v\in V(H)$, using
at most~$k$ operations from $S$. We prove that
\begin{itemize}
\item if $S=\{\ea\}$ or $S=\{\ea,\ed\}$, then \cdpe($S$) can be solved in polynomial time;
\item if $\{\vd\} \subseteq S\subseteq \{\ea,\ed,\vd\}$, then \cdpe($S$) is \NP-complete and \W-hard when parameterized by~$k$, even if $\delta\equiv 0$.
\end{itemize}
Together with known results by Cai and Yang and by Cygan, Marx, Pilipczuk,
Pilipczuk and Schlotter, our results completely classify the classical and
parameterized complexity of the \cdpe($S$) problem for all $S\subseteq
\{\ea,\ed,\vd\}$. We obtain the same classification for a natural variant of
the \cdpe($S$) problem on directed graphs, where the target is a weakly
connected digraph in which the difference between the in- and out-degree of
every vertex equals a prescribed value.
As an important implication of our results, we obtain polynomial-time
algorithms for the {\sc Eulerian Editing} problem and its directed variant. 
\end{abstract}


\section{Introduction}\label{sec:intro}
Graph modification problems play a central role in algorithmic graph theory,
partly due to the fact that they naturally arise in numerous practical
applications. A graph modification problem takes as input a graph $G$ and an
integer $k$, and asks whether $G$ can be modified into a graph
belonging to a prescribed graph class ${\cal H}$, using at most $k$ operations
of a certain type. The most common operations that are considered in this
context are edge additions (${\cal H}$-{\sc Completion}), edge deletions
(${\cal H}$-{\sc Edge Deletion}), vertex deletions (${\cal H}$-{\sc Vertex Deletion}),
and a combination of edge additions and edge deletions (${\cal H}$-{\sc
Editing}). The intensive study of graph modification problems has produced a
plethora of classical and parameterized complexity results (see
e.g.~\cite{BoeschST77,BurzynBD06,Ca96,CaiY11,CechlarovaS10,CrowstonGJY12,CyganMPPS14,DornMNW13,FroeseNN14,GoyalMPPS14,Golovach13,Golovach13a,HohnJM12,LesniakO86,LY80,MathiesonS12,MoserT09,NatanzonSS01}).

An undirected graph is Eulerian if it is connected and every vertex has even
degree,
while 
a directed graph is Eulerian if it is strongly
connected\footnote{Replacing ``strongly connected'' by ``weakly connected''
yields an equivalent definition of Eulerian digraphs, as it is well-known that
a balanced digraph is strongly connected if and only it is weakly connected
(see e.g.~\cite{CyganMPPS14}).} and balanced, i.e. the in-degree of every
vertex equals its out-degree. Eulerian graphs form a well-known graph class
both within algorithmic and structural graph theory. 
Several groups of authors have investigated the problem of deciding 
whether
a given undirected graph can be made Eulerian using a small number of operations. Boesch et
al.~\cite{BoeschST77} presented a polynomial-time algorithm for {\sc Eulerian
Completion}, and 
Cai and Yang~\cite{CaiY11} showed that the problems {\sc
Eulerian Vertex Deletion} and {\sc Eulerian Edge Deletion} are
\NP-complete~\cite{CaiY11}. When parameterized by~$k$, 
it is known that
{\sc Eulerian Vertex Deletion} is \W-hard~\cite{CaiY11}, while {\sc
Eulerian Edge Deletion} is fixed-parameter tractable~\cite{CyganMPPS14}. Cygan
et al.~\cite{CyganMPPS14} showed that the classical and parameterized
complexity results for {\sc Eulerian Vertex Deletion} and {\sc Eulerian Edge
Deletion} also hold for the directed variants of these problems.
Recently, Goyal et al.~\cite{GoyalMPPS14} improved the fixed-parameter tractability  results of Cygan et al.~\cite{CyganMPPS14} for the directed and undirected variants 
of {\sc Eulerian Edge
Deletion}. The same authors also proved that the {\sc Undirected Connected Odd Edge Deletion} problem, which asks whether it is possible to obtain a connected graph in which all vertices have odd degree by deleting at most $k$ edges, is fixed-parameter tractable when parameterized by $k$.

Another problem that can be seen as involving editing to an Eulerian multigraph
is the \textsc{Chinese Postman} problem, also known as the \textsc{Route
Inspection} problem~\cite{Kwan62}. In this problem a connected graph $G$,
together with an integer $k$, is given and the question is whether there exists
a closed walk in $G$ that contains all edges of $G$, but that has length at
most~$k$. In other words, can a total of at most $k$ copies of existing edges
be added to $G$ in order to modify $G$ into an Eulerian multigraph? Edmonds and
Johnson~\cite{EdmondsJ73} showed that both the undirected and directed variant
of this problem can be solved in polynomial time.

\subparagraph{Our Contribution}
We generalize, extend and complement known results on graph modification
problems dealing with Eulerian graphs and digraphs. 
The main contribution of
this paper consists of two non-trivial polynomial-time algorithms: one for
solving the {\sc Eulerian Editing} problem, and one for solving the directed
variant of this problem.
Given the aforementioned \NP-completeness result
for {\sc Eulerian Edge Deletion} and the fact that ${\cal H}$-{\sc Editing} is
\NP-complete for almost all natural graph classes ${\cal
H}$~\cite{BurzynBD06,NatanzonSS01}, we find it particularly interesting that
{\sc Eulerian Editing} turns out to be polynomial-time solvable. To the best of
our knowledge, 
the only other natural 
non-trivial graph class ${\cal H}$ for which ${\cal
H}$-{\sc Editing} is known to be polynomial-time solvable is the class of split
graphs~\cite{HammerS81}.

In fact, our polynomial-time algorithms are implications of two more general
results. In order to formally state these results, we need to introduce some
terminology. Let $\ea$, $\ed$ and $\vd$ denote the operations edge addition,
edge deletion and vertex deletion, respectively. For any set $S\subseteq
\{\ea,\ed,\vd\}$ and non-negative integer $k$, we say that a graph $G$ can be
{\em $(S,k)$-modified} into a graph $H$ if $H$ can be obtained from $G$ by
using at most $k$ operations from~$S$. We define the following problem for
every $S\subseteq \{\ea,\ed,\vd\}$:
\begin{center}
\begin{boxedminipage}{.99\textwidth}
\begin{tabular}{rl}
\textsc{\cdpe($S$):} & \textsc{Connected Degree Parity Editing$(S)$}\\
\textit{~~~~Instance:} & A graph $G$, an integer $k$\\
                       & and a function $\delta\colon V(G)\rightarrow\{0,1\}$.\\
\textit{Question:} & Can $G$ be $(S,k)$-modified into a connected graph $H$\\ 
                   & with $d_{H}(v)\equiv\delta(v)~(\bmod~2)$ for each $v\in V(H)$?
                   \end{tabular}
\end{boxedminipage}
\end{center}

Inspired by the work of Cygan et al.~\cite{CyganMPPS14} on directed Eulerian
graphs, we also study a natural directed variant of the \cdbe($S$) problem.
Denoting the in- and out-degree of a vertex $v$ in a digraph~$G$ by $d_G^-(v)$
and $d_G^+(v)$, respectively, we define the following problem for every
$S\subseteq\{\ea, \ed, \vd\}$:
\begin{center}
\begin{boxedminipage}{.99\textwidth}
\begin{tabular}{rl}
\textsc{\cdbe($S$):}   & \textsc{Connected Degree Balance Editing$(S)$}\\
\textit{~~~~Instance:} & A digraph $G$, an integer $k$ and\\
                       & a function $\delta\colon V(G)\rightarrow\mathbb{Z}$.\\
\textit{Question:}     & Can $G$ be $(S,k)$-modified into a weakly connected\\& digraph $H$
                        with $d_{H}^+(v)-d_{H}^-(v)=\delta(v)$ for each\\
& $v\in V(H)$?
\end{tabular}
\end{boxedminipage}
\end{center}

In Section~\ref{sec:undirected}, we prove that \cdpe($S$) can be solved in polynomial time when $S=\{\ea\}$  
and when $S=\{\ea,\ed\}$.
The first of these two results extends the result by Boesch et
al.~\cite{BoeschST77} on {\sc Eulerian Completion}
 and the second yields the first
polynomial-time algorithm for {\sc Eulerian Editing}, as these problems are
equivalent to \cdpe($\{\ea\}$) and \cdpe($\{\ea,\ed\})$, respectively, when we
set $\delta\equiv 0$. 
The complexity of the problem drastically changes when vertex deletion is allowed: we prove that for every subset $S\subseteq
\{\ea,\ed,\vd\}$ with $\vd\in S$, the \cdpe($S$) problem is
\NP-complete and \W-hard with parameter~$k$, even when $\delta\equiv 0$. 
This complements results by Cai and Yang~\cite{CaiY11} stating that \cdpe($S$)
is \NP-complete and \W-hard with parameter~$k$ when $S=\{\vd\}$
and $\delta\equiv 0$ or $\delta\equiv 1$. 
Our results, together with the aforementioned results due to Cygan et al.~\cite{CyganMPPS14}\footnote{The \FPT-results
by Cygan et al.~\cite{CyganMPPS14} only cover \cdpe($\{\ed\}$) and
\cdbe($\{\ed\}$) when $\delta\equiv 0$, but it can easily be seen that their
results carry over to \cdpe($\{\ed\}$) and \cdbe($\{\ed\}$) for any function
$\delta$.} and Cai and
Yang~\cite{CaiY11}, yield a complete classification of both the classical and the parameterized
complexity of \cdpe($S$) for all $S\subseteq \{\ea,\ed,\vd\}$; 
see the middle column of Table~\ref{t-thetable}. 

In Section~\ref{sec:directed}, we use different and more involved arguments to classify the classical and parameterized complexity of the \cdbe($S$) problem for all $S\subseteq \{\ea,\ed,\vd\}$. 
Interestingly, the classification we obtain for \cdbe($S$) turns out to be identical to the one we obtained for \cdpe($S$).  
In particular, our proof of the fact that \cdbe($S$) is polynomial-time solvable when $S=\{\ea\}$ and $S=\{\ea,\ed\}$ implies that the directed variants of {\sc Eulerian Completion} and {\sc Eulerian Editing} are not significantly harder than their undirected counterparts.
All results on \cdbe($S$) are summarized in the right column of Table~\ref{t-thetable}. 

\begin{table}[htb]
\begin{center}
\begin{tabular}{l|l|l}
$S$                    & $\;\;$ \cdpe($S$)                             & $\;\;$ \cdbe($S$)                           \\ \hline
$\ea,\ed$           & $\;\;$ \P                                             & $\;\;$ \P                                            \\
$\ea$                  & $\;\;$ \P                                             & $\;\;$  \P                                             \\
$\ed$                  & $\;\;$ \FPT~\cite{CyganMPPS14} &$\;\;$  \FPT~\cite{CyganMPPS14} \\
$\vd$                   & $\;\;$ \W-hard~\cite{CaiY11}$\;$               & $\;\;$ \W-hard~\cite{CyganMPPS14}\\
$\ea,\vd$            & $\;\;$ \W-hard                               & $\;\;$ \W-hard\\
$\ed,\vd$            & $\;\;$  \W-hard                              &$\;\;$  \W-hard \\
$\ea,\ed,\vd\;$     & $\;\;$ \W-hard                             & $\;\;$ \W-hard \\
\end{tabular}
\end{center}
\caption{A summary of the results for \cdpe($S$) and \cdbe($S$). All results are new except those for which a reference is given.
The number of allowed operations $k$ is the parameter in the parameterized results, and if a parameterized result is stated, then the corresponding problem is
\NP-complete.} 
\label{t-thetable}
\end{table}

We would like to emphasize that there are no obvious hardness reductions between the different problem variants. The parameter $k$ in the problem definitions represents the budget for all operations in total; adding a new operation to $S$ may completely change the problem, as there is no way of forbidding its use. Hence, our polynomial-time algorithms for \cdpe($\{\ea,\ed\}$) and \cdbe($\{\ea,\ed\}$) do not generalize the polynomial-time algorithms for \cdpe($\{\ea\}$) and \cdbe($\{\ea\}$), and as such require significantly different arguments. In particular, our main result, stating that {\sc Eulerian Editing} is polynomial-time solvable, is not a generalization of the fact that {\sc Eulerian Completion} is polynomial-time solvable and stands in no relation to the \FPT-result by Cygan et al.~\cite{CyganMPPS14} for {\sc Eulerian Edge Deletion}.

\medskip
\noindent
We end this section by mentioning two similar graph modification frameworks in the literature that formed a direct motivation for the framework defined in this paper.
Mathieson and Szeider~\cite{MathiesonS12} considered the {\sc Degree Constraint Editing($S$)} problem, which is that of testing whether a graph $G$ can be $(S,k)$-modified into a graph $H$ in which the degree of every vertex belongs to some list associated with that vertex;
recently some new results for this problem were obtained by Froese et al.~\cite{FroeseNN14} and Golovach~\cite{Golovach13a}.
Golovach~\cite{Golovach13} performed a similar study to that of Mathieson and Szeider~\cite{MathiesonS12}, but with the additional condition that the resulting graph must be connected.

\section{Preliminaries}\label{sec:prelim}

We consider finite graphs $G=(V,E)$ that may be undirected or directed; in the
latter case we will always call them digraphs. All our undirected graphs will
be without loops or multiple edges; in particular, this is the case for both the input and the output graph in every undirected problem we consider. 
Similarly, for every directed problem that we consider, we do not allow the input or output digraph to contain multiple arcs.
In our proofs we will also make use of {\em directed multigraphs}, which are digraphs that are permitted to have multiple arcs.

We denote an edge between two vertices $u$ and $v$ in a graph by $uv$.  We denote an arc between
two vertices $u$ and~$v$ by $(u,v)$, where $u$ is the \emph{tail} of $(u,v)$
and $v$ is the \emph{head}. The disjoint union of two graphs $G_1$ and $G_2$ is denoted $G_1+G_2$.
The complete graph on $n$ vertices is denoted $K_n$ and the complete bipartite graph with classes of size $s$ and $t$
is denoted~$K_{s,t}$.

Let $G=(V,E)$ be a graph or a digraph.  Throughout the paper we assume that
$n=|V|$ and $m=|E|$.  For $U\subseteq V$, we let $G[U]$  be the graph (digraph)
with vertex set $U$ and an edge (arc) between two vertices $u$ and~$v$ if and
only if this is the case in $G$; we say that $G[U]$ is \emph{induced by}~$U$.
We write $G-U=G[V\setminus U]$.  For $E'\subseteq E$, we let $G(E')$ be the
graph (digraph) with edge (arc) set~$E'$ whose vertex set consists of the
end-vertices of the edges in $E'$; we say that $G(E')$ is \emph{edge-induced
by}~$E'$.  Let $S$ be a set of (ordered) pairs of vertices of $G$.  We let
$G-S$ be the graph (digraph) obtained by deleting all edges (arcs) of $S\cap E$
from~$G$, and we let $G+S$ be the graph (digraph) obtained by adding all edges
(arcs) of $S\setminus E$ to $G$.  We may write $G-e$ or $G+e$ if $S=\{e\}$.

Let $G=(V,E)$ be a graph.  A {\em component} of $G$ is a maximal connected subgraph
of~$G$.  The {\em complement}  of $G$  is the graph $\overline{G}=(V,\overline{E})$
with vertex set~$V$ and an edge between two distinct vertices $u$ and $v$ if and only
if~$uv\notin E$.   For a vertex $v\in V$, we
let $N_G(v)=\{u \mid uv\in E\}$ denote its \emph{(open) neighbourhood}.  The
\emph{degree} of  $v$ is denoted $d_G(v)=|N_G(v)|$. The graph $G$  is {\em even}
if all its vertices have even degree, and it is {\em Eulerian} if it is even
and connected.  We say that a set $D\subseteq E$ is an {\em edge cut} in $G$
if~$G$ is connected but $G-D$ is not.  An edge cut of size~$1$ is called a {\em
bridge} in~$G$.

A \emph{matching} of a graph $G$ is a set of edges, in which no two edges have
a common end-vertex; it is called a {\em maximum} matching if its number of
edges is maximum over all matchings of $G$.  We need the following lemma due to
Micali and Vazirani.

\begin{lemma}[\cite{MicaliV80}]\label{l-mic}
A maximum matching of an $n$-vertex graph can be found in $O(n^{5/2})$ 
time.
\end{lemma}

Let $G=(V,E)$ be a digraph.  If $(u,v)$ is an arc, then $(v,u)$ is the {\em
reverse} of this arc.  For a subset $F\subseteq E$, we let $F^R=\{(u,v) | (v,u)
\in F\}$ denote the set of arcs whose reverse is in~$F$.  The \emph{underlying}
graph of~$G$ is the undirected graph with vertex set $V$ where two vertices
$u,v\in V$ are adjacent if and only if $(u,v)$ or $(v,u)$ is an arc in $G$.  We
say that $G$ is \emph{(weakly) connected} if its underlying graph is connected.
A {\em component} of~$G$ is a connected component of its underlying graph.  An
arc $a\in E$ is a \emph{bridge} in $G$ if it is a bridge in the underlying
graph of $G$.  A vertex $u$ is an \emph{in-neighbour} or \emph{out-neighbour}
of a vertex~$v$ if $(u,v)\in E$ or $(v,u)\in E$, respectively. Let
$N_G^-(v)=\{u\mid (u,v)\in E\}$ and $N_G^+(v)=\{u\mid (v,u)\in E\}$, where we
call $d_G^-(v)=|N_G^-(v)|$ and $d_G^+(v)=|N_G^+(v)|$ the \emph{in-degree} and
\emph{out-degree} of $v$, respectively. 
A vertex $v\in V$ is {\em balanced} if $d_G^+(v)=d_G^-(v)$, 
or equivalently, its {\em degree balance} $d_G^+(v)-d_G^-(v)=0$.
Recall that~$G$ is {\em Eulerian} if it
is connected and {\em balanced}, that is, the out-degree of every vertex is
equal to its in-degree. 

Let $G=(V,E)$ be a graph and let $T\subseteq V$. A subset $J \subseteq E$ is a
\emph{$T$-join} if the set of odd-degree vertices in $G(J)$ is precisely $T$.
If $G$ is connected and $|T|$ is even then $G$ has at least one $T$-join. In
Section~\ref{sec:undirected} we need to find a {\em minimum} $T$-join, that is,
one of minimum size. We use the following result of Edmonds and Johnson~\cite{EdmondsJ73} to do
so.

\begin{lemma}[\cite{EdmondsJ73}]\label{lem:t-join}
Let $G=(V,E)$ be  a graph, and let $T\subseteq V$.  Then a minimum $T$-join (if
one exists) can be found in $O(n^3)$ time. 
\end{lemma}
Lemma~\ref{lem:t-join} was used by Cygan et al.~\cite{CyganMPPS14} to solve
${\cal H}$-{\sc Edge Deletion} in polynomial time when~${\cal H}$ is the class
of even graphs. It would  immediately yield a polynomial-time algorithm for
\cdpe($\{\ed\}$) if we dropped the connectivity condition.

We need a variant of Lemma~\ref{lem:t-join} for digraphs in
Section~\ref{sec:directed}.  Let $G=(V,E)$ be a directed multigraph and let $f:
T\rightarrow \mathbb{Z}$ be a function for some $T\subseteq V$.  A multiset
$E'\subseteq E$ with $T\subseteq V(G(E'))$ is a {\em directed $f$-join} in $G$
if the following two conditions hold: 
$d_{G(E')}^+(v)-d_{G(E')}^-(v)=f(v)$ for
every $v \in T$ and $d_{G(E')}^+(v)-d_{G(E')}^-(v)=0$ for every $v\in
V(G(E'))\setminus T$.  A directed $f$-join is {\em minimum} if it has minimum
size.  The next lemma was used by Cygan et al.~\cite{CyganMPPS14} to solve
${\cal H}$-{\sc Edge Deletion} in polynomial time when ${\cal H}$ is the class
of balanced digraphs; it would also yield a polynomial-time algorithm for
\cdbe($\{\ed\}$) if we dropped the connectivity condition. 

\begin{lemma}[\cite{CyganMPPS14}]\label{lem:dir-t-join}
Let $G=(V,E)$ be a directed multigraph and $f: T \rightarrow \mathbb{Z}$ be  a
function for some $T\subseteq V$.  A minimum directed $f$-join $F$ (if one
exists) can be found in $O(nm\log n \log \log m)$ time. Moreover,~$F$ consists
of mutually arc-disjoint directed paths from vertices~$u$ with $f(u)>0$ to
vertices~$v$ with $f(v)<0$.
\end{lemma}

\section{Connected Degree Parity Editing}\label{sec:undirected}

Let $S\subseteq \{\ea,\ed,\vd\}$.  
In Section~\ref{s-polyun} we will show that
\cdpe($S$) is polynomial-time solvable if $S=\{\ea\}$ or $S=\{\ea,\ed\}$ and
in Section~\ref{s-wun} we will show 
that it is \NP-complete and \W-hard with parameter~$k$ if $\vd\in S$.

\subsection{The Polynomial-Time Solvable Cases}\label{s-polyun}
First, let $\{\ea\}\subseteq S \subseteq \{\ea,\ed\}$.  Let $(G,\delta,k)$ be an
instance of \cdpe($S$) with $G=(V,E)$.  Let~$A$ be a set of edges not in $G$,
and let $D$ be a set of edges in~$G$, with $D=\emptyset$ if $S=\{\ea\}$.  We
say that $(A,D)$ is a {\em solution} for $(G,\delta,k)$ if its {\em size}
$|A|+|D|\leq k$, the congruence $d_H(u)\equiv\delta(u)~(\bmod~2)$ holds for every vertex~$u$ and
the graph $H=G+A-D$ is connected; if $H$ is
not connected then $(A,D)$ is a {\em semi-solution} for $(G,\delta,k)$. If
$S=\{\ea\}$ we may denote the solution by $A$ rather than $(A,D)$ (since
$D=\emptyset$).  We consider the optimization version for \cdpe$(S)$.  The
input is a pair $(G,\delta)$, and we aim to find the minimum $k$ such that
$(G,\delta,k)$ has a solution (if one exists).  We call such a solution {\em
optimal} and denote its size by $\opts(G,\delta)$.  We say that a
(semi)-solution for $(G,\delta,k)$ is also a (semi)-solution for $(G,\delta)$.
If $(G,\delta,k)$ has no solution for any value of $k$, then $(G,\delta)$ is a
{\em no-instance} of  \cdpe$(S)$ and $\opts(G,\delta)=+\infty$.

Let $T=\{v\in V\;|\;d_G(v)\not\equiv \delta(v)~(\bmod~2)\}$.  Define
$G_S=K_n$ if $S=\{\ea,\ed\}$ and $G_S=\overline{G}$ if $S=\{\ea\}$.  Note that
if $S=\{\ea\}$ then $G_S$ contains no edges of $G$, so in this case any
$T$-join in $G_S$ can only contain edges in $E(\overline{G})$. The following
key lemma is an easy observation.

\begin{sloppypar}
\begin{lemma}\label{lem:struct-undir}
Let $\{\ea\}  \subseteq S \subseteq \{\ea,\ed\}$. Let $(G,\delta)$ be an instance of
\cdpe$(S)$ and~$A \subseteq E(\overline{G})$, $D \subseteq E(G)$.  Then
$(A,D)$ is a semi-solution of \cdpe$(S)$ if and only if $A \cup D$ is a
$T$-join in~$G_S$.
\end{lemma}
\end{sloppypar}

We extend the result of Boesch et al.~\cite{BoeschST77} for $\delta\equiv 0$ to arbitrary $\delta$. Our proof
is based around similar ideas but we also had to do some further analysis. The
main difference in the two proofs is the following. If $\delta \equiv 0$ then
none of the added edges in a solution will be a bridge in the modified graph
(as the number of vertices of odd degree in a graph is always even). However
this is no longer true for arbitrary $\delta$ and extra arguments are needed.

\begin{theorem}\label{thm:add-undi}
Let $S=\{\ea\}$. Then \cdpe$(S)$ can be solved in $O(n^3)$ time.
\end{theorem} 

\begin{proof}
Let $S=\{\ea\}$ and let $(G,\delta)$ be an instance of \cdpe$(S)$.  We first
use Lemma~\ref{lem:t-join} to check in $O(n^3)$ time whether $G_S$ has a
$T$-join.  If not then $(G,\delta)$ has no semi-solution by
Lemma~\ref{lem:struct-undir}, and thus no solution either.  We may therefore
assume that~$|T|$ is even and~$F$ is a minimum $T$-join in~$G_S$. (Recall that
Lemma~\ref{lem:t-join} states that we can find $F$ in $O(n^3)$ time if it
exists.) We also assume that either~$T \neq \emptyset$ or $G$ is not connected,
otherwise the trivial solution $A=\emptyset$ is clearly optimal.  Let~$p$ be
the number of components of~$G$ that do not contain any vertex of~$T$ and
let~$q$ be the number of components of~$G$ that contain at least one vertex
of~$T$. We will prove the following series of statements.

\begin{itemize}
\item $(G,\delta)$ is a no-instance if $p=2,q=0$ and $G=K_1+K_t$ for $t\geq 1$.
\item $\opts(G,\delta)=4$ if $p=2,q=0$ and $G=K_s+K_t$ for $s,t\geq 2$.
\item $\opts(G,\delta)=3$ if $p=2,q=0$ and~$G$ has a component that is not complete.
\item $\opts(G,\delta)=p$ if $p\geq 3,q=0$.
\item $\opts(G,\delta)=\max\{|F|,p+q-1,p+\frac{1}{2}|T|\}$ if $q>0$.
\end{itemize}

We split our proof into two parts depending on the value of~$q$.

\medskip
\noindent
\displaycase{Case 1:} $q=0$.\\
In this case $T=\emptyset$, so by Lemma~\ref{lem:struct-undir} for any
semi-solution~$A$, every vertex in $G_S(A)$ must have even degree in $G_S(A)$.
In other words, every vertex of~$G$ must be incident to an even number of edges
in~$A$. Since $T=\emptyset$, we assumed above that~$G$ was disconnected, so $p\geq 2$
and any solution~$A$ must be non-empty. This means that $G_S(A)$ must contain a
cycle, so $\opts(G,\delta)\geq 3$. Recall that $G_S(A)$ is
a subgraph of $\overline{G}$.

Suppose $p=2$. If $G=K_1+K_t$ for $t\geq 2$ then $\overline{G}=K_{1,t}$, which
does not contain a cycle. Therefore $(G,\delta)$ is a no-instance in this case.
If $G=K_s+K_t$ for $s,t\geq 2$ then $\overline{G}=K_{s,t}$, which contains no
cycles of length 3. Therefore $\opts(G,\delta)\geq 4$ in this case.  Indeed, if
$u,v$ are vertices in the~$K_s$ component of~$G$ and $u',v'$ are vertices in the~$K_t$
component, then $A=\{uu',u'v,vv',v'u\}$ is a solution of size 4 and this
solution must therefore be optimal.  Finally, suppose~$G$ contains exactly two
components, at least one of which is not a clique. Let~$x,y$ be non-adjacent
vertices in this component and let~$z$ be a vertex in the other component. Then
$A=\{xy,yz,zx\}$ is a solution of size~3, which must therefore be optimal.

Finally, suppose that $p \geq 3$. Since $G+A$ must be connected for any
solution~$A$, every component in~$G$ must contain at least one vertex
incident to an edge of~$A$. By Lemma~\ref{lem:struct-undir}, this vertex must
be incident to an even number of edges of~$A$, meaning that it must be incident
to at least two such edges. Therefore $\opts(G,\delta) \geq p$.  Indeed, if we
choose vertices $v_1,\ldots,v_p$, one from each component of~$G$ then
$A=\{v_1v_2,v_2v_3,\ldots,v_{p-1}v_p,v_pv_1\}$ is a solution of size~$p$, which
is therefore optimal.

This concludes the $q=0$ case.

\medskip
\noindent
\displaycase{Case 2:} $q>0$.\\
In this case $T\neq\emptyset$. We first show that $\opts(G,\delta)\geq
\max\{|F|,p+q-\nobreak 1,\allowbreak p+\frac{1}{2}|T|\}$. Since~$F$ is a minimum $T$-join in~$G_S$,
Lemma~\ref{lem:struct-undir} implies that $\opts(G,\delta)\geq |F|$. Since $G$
has $p+q$ components, any solution $A$ must contain at least $p+q-1$ edges to
ensure that $G+A$ is connected, so $\opts(G,\delta)\geq p+q-1$. Finally, let
$G_1,\ldots,G_p$ be the components of~$G$ that do not contain any vertices
of~$T$. If $A$ is a solution then every component~$G_i$ must contain a vertex
incident to some edge in~$A$.  By Lemma~\ref{lem:struct-undir}, this vertex
must be incident to an even number of edges of $A$, meaning that it must be
incident to at least two such edges.  By Lemma~\ref{lem:struct-undir}, every
vertex of~$T$ must be incident to some edge in~$A$. Therefore~$A$ must contain
at least $p+\frac{1}{2}|T|$ edges, so $\opts(G,\delta)\geq p+\frac{1}{2}|T|$.

Next we show that we can always construct a solution of size
$\max\{|F|,\allowbreak p+q-1,p+\frac{1}{2}|T|\}$. To do this, we try to replace edges of~$F$
in such way that~$F$ remains a minimum $T$-join in~$G_S$, but the number of
components in $G+F$ is reduced. After we have finished this process, if $G+F$
is connected then setting $A=F$ gives a solution of size~$|F|$, which is
therefore optimal. Otherwise, we will be able to use the structure of~$F$ to
construct a solution of size either $p+q-1$ or $p+\frac{1}{2}|T|$.

Consider the graph $G_S(F)$. Since~$F$ is a minimum $T$-join, $G_S(F)$ cannot
contain any cycles (otherwise the edges in the cycle could be removed from~$F$
to give a smaller $T$-join).  We claim that~$G_S(F)$ does not contain a path of
length $\geq 3$.  Suppose, for contradiction, that there is such a path with
edge set~$P$ and end-vertices~$u$ and~$v$.  Note that~$u$ and~$v$ are in the
same component of $G+F$. Since $G+F$ is not connected (otherwise $A=F$ would be
an optimal solution of size $|F|$), there must be a vertex $x \in V(G)$ which
is in a different component of $G+F$ from the one containing~$u$ and~$v$. In this
case $ux,xv \in E(G_S)$.  Let~$F'=F\setminus P \cup \{ux,xv\}$. Then~$F'$ is
also a $T$-join in $G_S$, since the degree parity of any vertex in $G+F'$ is
the same as its degree parity in $G+F$. However, $|F'|<|F|$, which contradicts
the fact that~$F$ is a minimum $T$-join.  Therefore $G_S(F)$ must be a forest
that contains no paths of length 3. In other words $G_S(F)$ is a forest of
stars.

Now suppose that $uv,u'v'\in F$, such that~$uv$ is not a bridge in $G+F$ and
the vertices~$u$ and~$u'$ are in different components of $G+F$.
Let~$F'=F\setminus\{uv,u'v'\}\cup\{u'v,uv'\}$. Then~$F'$ is also a minimum
$T$-join in~$G_S$. However, $G+F'$ has one component less than $G+F$.  Indeed,
since~$uv$ is not a bridge in $G+F$, the vertices $u,u',v,v'$ must all be in
the same component of $G+F'$. Therefore, if such edges $uv,u'v'\in F$ exist, we
replace~$F$ by~$F'$. We do this exhaustively until no further such pairs of
edges exist. At this point either every edge in~$F$ must be a bridge or every
edge in~$F$ is in the same component of $G+F$. We consider these possibilities
separately.

First suppose that every edge in~$F$ is a bridge. Choose $uv \in F$ and let
$G_1,\ldots,G_k$ be the components of $G+F$, with $u,v \in V(G_1)$. Note that
since every edge in~$F$ is a bridge, $k=p+q-|F|$. Now let $v_i \in V(G_i)$ for
$i \in \{2,\ldots,k\}$. Let $A=F$ if $k=1$ and $A=F \setminus \{uv\} \cup
\{uv_2,v_2v_3,\ldots,v_{k-1}v_k,v_kv\}$ otherwise. Now every vertex in $G+A$
has the same degree parity as in $G+F$, so~$A$ is a $T$-join in~$G_S$.  The
graph $G+A$ is connected, so~$A$ is a solution.  However,
$|A|=|F|-1+p+q-|F|=p+q-1$.  Therefore~$A$ is an optimal solution.

We may now assume that every edge in~$F$ is in the same component of $G+F$. If
$G+F$ is connected, then $A=F$ is a solution of size $|F|$ and is therefore
optimal, so we may assume that $G+F$ is not connected. Suppose $uv,vw \in F$. Then
$uw \in E(G)$, otherwise we could replace~$uv,vw$ in $F$ by $uw$ to get a
smaller $T$-join in $G_S$. Suppose that $uv,vw$ do not form a cut-set in $G+F$.
In other words, we suppose that $u$ and $v$ are in the same component of
$G+F\setminus\{uv,vw\}$. Let~$x$ be a vertex in a different component of $G+F$
from the one containing $u,v,w$. Then $ux,xw \in E(G_S)$.  Let
$F'=F\setminus\{uv,vw\}\cup\{ux,xw\}$. Then $F'$ must also be a minimum
$T$-join in~$G_S$. However, $G+F'$ has one less component than $G+F$.
Indeed,~$x$ is in the same component of $G+F'$ as $u,v,w$. In this case we may
replace~$F$ by~$F'$. Again, we apply this replacement exhaustively until it can
no longer be applied. This process ends when either $G+F$ becomes connected (in
which case $A=F$ is an optimal solution of size $|F|$) or, for every pair of
edges of the form $uv,vw \in F$, we find that $\{uv,vw\}$ is a cut-set in
$G+F$. We may assume the latter is the case.

Now suppose $uv,vw \in F$. Consider the component $C$ of
$G+F\setminus\{uv,vw\}$ containing $v$. We claim that $C$ contains no vertices
of $T$. Suppose, for contradiction, that $x \in T \cap C$ ($x$ is not
necessarily distinct from $v$). Then by Lemma~\ref{lem:struct-undir},~$x$ must
be the end-vertex of some edge in $F\setminus\{uv,vw\}$, say $xy$ (again $y$ is
not necessarily distinct from $v$). Note that $x$ and~$y$ are in the same
component of $G+F\setminus\{uv,vw\}$, which is different from the component
containing $u$ and $w$. Let $F'=F\setminus\{xy,uv,vw\}\cup\{ux,yw\}$, then~$F'$
is also a $T$-join in $G_S$, but $|F'|=|F|-1$, contradicting the minimality of
$F$. Therefore $C$ must be one of the $p$ components of $G$ that contain no
vertices of~$T$.

Now $G_S(F)$ contains $\frac{1}{2}|T|$ paths and $|F|$ edges, so we can
decompose $G_S(F)$ into $|T|-|F|$ paths of length 1 and $|F|-\frac{1}{2}|T|$
paths of length~2. We can do this in such a way that the ends of each path lie
in $T$.  Also, by the arguments above, the middle vertex of every path of length~2
lies in a different one of those~$p$ components of~$G$ that do not contain any
vertices of $T$.  Let $G_0,G_1,\ldots,G_k$ be the components of $G+F$ such that
$G_0$ is the only component containing vertices of $T$. Note that
$k=p-(|F|-\frac{1}{2}|T|)$.  Let $v_i \in V(G_i)$ for $i \in \{1,\ldots,k\}$.
Choose $uv \in F$ and let $A=F
\setminus\{uv\}\cup\{uv_1,v_1v_2,\ldots,v_{k-1}v_k,v_kv\}$. Then every vertex
in $G+A$ has the same degree parity as in $G+F$ and the graph $G+A$ is
connected, so $A$ is a solution. Furthermore, $|A|= |F|+p-(|F|-\frac{1}{2}|T|)=
p+\frac{1}{2}|T|$, so $A$ is an optimal solution.  This concludes the proof of
Case 2.

\medskip
\noindent
Recall that a minimum $T$-join in~$G_S$ can be found in $O(n^3)$ time by
Lemma~\ref{lem:t-join}, so the value of $\opts(G,\delta)$ can be computed in
$O(n^3)$ time. Note that the constructive proofs for Cases~1 and 2 can be
turned into
$O(nm)$ time algorithms, so an optimal solution $A$ can also be found in
$O(n^3)$ time.
\end{proof}

We are now ready to present the main result of this section. 
Proving this result requires significantly different arguments than the ones used in the proof of Theorem~\ref{thm:add-undi}.
Let $S=\{\ea,\ed\}$ and let $(G,\delta)$ be an instance of \cdpe($S$). If~$F$
is a $T$-join in $G_S=K_n$, let $D=F \cap E(G)$ and $A=F \setminus D$. Then by
Lemma~\ref{lem:struct-undir}, $(A,D)$ is a semi-solution.  Note that if~$F$ is
a minimum $T$-join in~$G_S$ then it is a matching in which every vertex of~$T$
is incident to precisely one edge of~$F$, so $|F|=\frac{1}{2}|T|$. We will show
how this allows us to calculate $\opts(G,\delta)$ directly from the structure
of~$G$, without having to find a $T$-join. We will also show that there are
only trivial no-instances for this problem,
namely when~$|T|$ is odd or~$G$ contains only two vertices.

\begin{theorem}\label{thm:edit-undir}
Let $S=\{\ea,\ed\}$. Then \cdpe$(S)$ can be solved in $O(n+m)$ time and an
optimal solution (if one exists) can be found in $O(n^3)$ time.
\end{theorem}

\begin{proof}
Let $S=\{\ea,\ed\}$ and let $(G,\delta)$ be an instance of \cdpe$(S)$.  By
Lemma~\ref{lem:struct-undir}, we may assume that~$|T|$ is even, otherwise
$(G,\delta)$ is a no-instance.  If $G=K_2$ and $T=V(G)$, or $G=K_1+K_1$ and
$T=\emptyset$, then $(G,\delta)$ is a no-instance.  
If $G=K_2$ and $T=\emptyset$ then, trivially, $\opts(G,\delta)=0$, and if
$G=K_1+K_1$ and $T=V(G)$ then $\opts(G,\delta)=1$.
To avoid these trivial
instances, we therefore assume that~$G$ contains at least three vertices.
Under these assumptions we will show that $\opts(G,\delta)$ is always finite and give exact formulas for the value of $\opts(G,\delta)$. 
Let~$p$ be the number of components of~$G$ that do not contain
any vertex of~$T$ and let~$q$ be the number of components of~$G$ that contain
at least one vertex of~$T$.
We  prove the following series of statements.

\begin{itemize}
\item $\opts(G,\delta)=0$ if $p=1,q=0$,
\item $\opts(G,\delta)=\max\{3,p\}$ if $p \geq 2,q=0$,
\item $\opts(G,\delta)=\frac{1}{2}|T|+1$ if $p=0,q=1$,  
$G[T]=K_{1,r}$, for some $r\geq 1$, and each edge of~$G[T]$ is a bridge of $G$,
\item $\opts(G,\delta)=\max\{p+q-1,p+\frac{1}{2}|T|\}$ in all other cases.
\end{itemize}

\smallskip
\noindent
Note that if $p=1,q=0$, then the first statement applies and the trivial
solution $(A,D)=(\emptyset,\emptyset)$ is optimal.
We now consider the remaining three cases separately.

\medskip
\noindent
\displaycase{Case 1:} {\em$p \geq 2$ and $q=0$.}\\
Then $T=\emptyset$, so by Lemma~\ref{lem:struct-undir} for any
semi-solution~$(A,D)$, every vertex in $G_S(A\cup D)$ must have even degree in
$G_S(A \cup D)$.  In other words, every vertex of~$G$ must be incident to an
even number of edges in~$A\cup D$.  Since $p \geq 2$, the graph~$G$ is
disconnected, so any solution $(A,D)$ is non-empty.  This means that
$G_S(A\cup D)$ must contain a cycle, so $\opts(G,\delta)\geq\nobreak 3$ if a solution
exits.
Suppose $p=2$. As~$G$ has at least three vertices, it contains a component containing an edge~$xy$. Let $z$ be a vertex in its
other component. We set $A=\{xz,yz\}$ and $D=\{xy\}$ to obtain a solution for $(G,\delta)$.
Since $|A|+|D|=3$, this solution is optimal. Suppose $p \geq\nobreak 3$. Since $G+A-D$ must be connected for any
solution~$(A,D)$, every component in~$G$ must contain at least one vertex
incident to an edge of~$A$. By Lemma~\ref{lem:struct-undir}, this vertex must
be incident to an even number of edges of~$A \cup D$, meaning that it must be
incident to at least two such edges. Therefore $\opts(G,\delta) \geq p$.
Indeed, if we choose vertices $v_1,\ldots,v_p$, one from each component of~$G$,
then setting $A=\{v_1v_2,v_2v_3,\ldots,v_{p-1}v_p,v_pv_1\}$ and $D=\emptyset$
gives a solution of size~$p$, which is therefore optimal.
This concludes Case~1.

\medskip
\noindent
\displaycase{Case 2:} {\em $p=0,q=1$, $G[T]=K_{1,r}$ for some $r\geq 1$ and each edge
of~$G[T]$ is a bridge of $G$.}\\
Then $G$ is connected.  Let $v_0$ be the central vertex of the star and
let $v_1,\ldots,v_r$ be the leaves.  By Lemma~\ref{lem:struct-undir}, in any
semi-solution $(A,D)$, every vertex of $T$ must be incident to an odd number of
edges in $A \cup D$, so $\opts(G,\delta) \geq \frac{1}{2}|T|$. Suppose $(A,D)$
is a semi-solution of size $|A|+|D|=\frac{1}{2}|T|$. Then $A \cup D$ must be a
matching with each edge joining a pair of vertices of $T$. However, then
$v_0v_i \in A \cup D$ for some~$i$. Since $v_0v_i \in E(G)$, we must have
$v_0v_i \in D$. However, since $v_0v_i$ is a bridge of $G$, $v_0$ and $v_i$
must then be in different components of $G+A-D$, so $G+A-D$ is not connected
and $(A,D)$ is not a solution.  Therefore $\opts(G,\delta) \geq
\frac{1}{2}|T|+1$.

Next we show how to find a solution of size $\frac{1}{2}|T|+1$.  Since~$|T|$ is
even, $r$ must be odd. First suppose that $r=1$. Since~$G$ is connected and
$v_0v_1$ is a bridge, $G \setminus \{v_0v_1\}$ has exactly two components.
Since~$G$ contains at least three vertices, one of these components contains
another vertex $x$. Without loss of generality assume $xv_0 \in E(G)$, in which
case  $xv_1 \notin E(G)$.
Then setting $A=\{xv_1\}$ and $D=\{xv_0\}$ gives a
semi-solution.  Since $x,v_0,v_1$ are all in the same component of $G+A-D$, the
graph $G+A-D$ must be connected, so $(A,D)$ is a solution. Since
$|A|+|D|=2=\frac{1}{2}|T|+1$, this solution is optimal.  
Now suppose $r \geq 3$.
Let $A=\{v_1v_2,v_2v_3\}\cup\{v_{2i}v_{2i+1}\;|\;2\leq i\leq
\frac{1}{2}(r-1)\}$ and $D=\{v_0v_2\}$. Then $(A,D)$ is a semi-solution and since
$v_0,\ldots,v_r$ are all in the same component of $G+A-D$, we find that  $(A,D)$ is a
solution. Since $|A|+|D|=2+\frac{1}{2}(r-1)-1+1=\frac{1}{2}|T|+1$, this solution is
optimal.
This concludes Case 2.

\medskip
\noindent
\displaycase{Case 3:} {\em $q\geq 1$ and Case~2 does not hold.}\\
Then $T\neq\emptyset$.
Let $G_1,\ldots,G_p$ be the components of~$G$ without vertices of~$T$  and let
$G'=G-V(G_1)\cup\cdots\cup V(G_p)$. Note that $G'=G$ if $p=0$ and that~$G'$ is not
the empty graph, as $q>0$.  Choose $v_i \in V(G_i)$ for $i\in\{1,\ldots,p\}$.

We first show that $\opts(G,\delta)\geq \max\{p+q-1,p+\frac{1}{2}|T|\}$.
Since~$G$ has $p+q$ components, any solution $(A,D)$ must contain at least
$p+q-1$ edges in $A$ to ensure that $G+A-D$ is connected, so
$\opts(G,\delta)\geq p+q-1$. If $(A,D)$ is a solution then every
component~$G_i$ must contain a vertex incident to some edge in~$A$.  By
Lemma~\ref{lem:struct-undir}, this vertex must be incident to an even number of
edges of $A \cup D$, meaning that it must be incident to at least two such
edges. By Lemma~\ref{lem:struct-undir}, every vertex of~$T$ must be incident to
some edge in~$A \cup D$. Therefore~$A \cup D$ must contain at least
$p+\frac{1}{2}|T|$ edges, so $\opts(G,\delta)\geq p+\frac{1}{2}|T|$.

We now show how to find a solution of size $\max\{p+q-1,p+\frac{1}{2}|T|\}$. We
start by finding a maximum matching~$M$ in~$\overline{G[T]}$. Let~$U$ be the
set of vertices in~$T$ that are not incident to any edge in~$M$. We divide the
argument into two cases, depending on the size of $U$.

\medskip
\noindent
\displaycase{Case 3a:} $U=\emptyset$.\\
In this case, by Lemma~\ref{lem:struct-undir}, setting $A=M$ and $D=\emptyset$
gives a semi-solution. Now suppose that $uv,u'v'\in M$, such that~$uv$ is not a
bridge in $G+M$ and the vertices~$u$ and~$u'$ are in different components of
$G+M$.  Let~$M'=M\setminus\{uv,u'v'\}\cup\{u'v,uv'\}$. Then~$M'$ is also a
maximum matching in~$\overline{G[T]}$. However, $G+M'$ has one component less
than $G+M$.  Indeed, since~$uv$ is not a bridge in $G+M$, the vertices
$u,u',v,v'$ must all be in the same component of $G+M'$. Therefore, if such
edges $uv,u'v'\in M$ exist, we replace~$M$ by~$M'$. We do this exhaustively
until no further such pairs of edges exist. At this point either every edge
in~$M$ is a bridge
in $G+M$
or every edge in~$M$ is in the same component of $G+M$.  We
consider these possibilities separately.

First suppose that every edge in~$M$ is a bridge
in $G+M$. 
Choose $uv \in\nobreak M$ and let
$Q_1,\ldots,Q_k$ be the components of $G+M$, with $u,v \in V(Q_1)$.  Note that
since every edge in~$M$ is a bridge, $k=p+q-|M|$. Now let $x_i \in V(Q_i)$ for
$i \in \{2,\ldots,k\}$. Let $D= \emptyset$ and let $A=M$ if $k=1$ and $A=M
\setminus \{uv\} \cup \{ux_2,x_2x_3,\ldots,x_{k-1}x_k,x_kv\}$ otherwise.  Now
every vertex in $G+A-D$ has the same degree parity as in $G+M$, so~$(A,D)$ is a
semi-solution by Lemma~\ref{lem:struct-undir}. The graph $G+A-D$ is connected,
so~$(A,D)$ is a solution. As  $|A|+|D|=|M|-1+p+q-|M|+0=p+q-1$, we find that $(A,D)$ is an optimal solution.

Now suppose that every edge in~$M$ is in the same component of $G+M$. Note that
$G_1,\ldots,G_p$ are the remaining components of $G+M$.  Choose $uv \in M$. Let
$D= \emptyset$ and let $A=M$ if $p=0$ and $A=M \setminus \{uv\} \cup
\{uv_1,v_1v_2,\ldots,v_{p-1}v_p,v_pv\}$ otherwise.  Then every vertex in
$G+A-D$ has the same parity as in $G+M$ and $G+A-D$ is connected, so by
Lemma~\ref{lem:struct-undir} $(A,D)$ is a solution. Since
$|A|+|D|=\frac{1}{2}|T|-1+p+1=p+\frac{1}{2}|T|$, this solution is optimal.
This concludes Case 3a.

\medskip
\noindent
\displaycase{Case 3b:} $U\neq\emptyset$.\\
Note that~$z=|U|$ must be even since~$|T|$ is even. Every pair of vertices
in~$U$ must be non-adjacent in~$\overline{G}$, as otherwise $M$ would not be
maximum. Therefore $G[U]$ is a
clique. Let $U=\{u_1,\ldots,u_z\}$.

We claim that $Q=G'+M$ is connected. Clearly every vertex of the clique~$U$
must be in the same component of $Q=G'+M$. Suppose for contradiction that $Q_1$
is a component of~$Q$ that does not contain~$U$. Then~$Q_1$ must contain
some edge $w_1w_2 \in M$. However, in this case $M'=M \setminus\{w_1w_2\}
\cup \{u_1w_1,u_2w_2\}$ is a larger matching in $\overline{G[T]}$ than $M$,
which contradicts the maximality of $M$. Therefore~$Q$ is connected.

Let $M'=\{u_1u_2,u_3u_4,\ldots,u_{z-1}u_z\}$. If $z \geq 4$ then since $U$ is a
clique, $G'+M-M'$ is connected. If $p=0$ set $A=M$ and $D=M'$. If $p>0$
set $A=M \cup \{u_1v_1,v_1v_2,\ldots,v_{p-1}v_p,v_pu_2\}$ and $D=M'
\setminus \{u_1u_2\}$. Then $G+A-D$ is connected, so $(A,D)$ is a solution by
Lemma~\ref{lem:struct-undir}. This solution has size $|A|+|D|=p+\frac{1}{2}|T|$,
so it is optimal.

Now suppose that $z\leq 3$. Then $z=2$. If $p>0$, let $A=M \cup
\{u_1v_1,\allowbreak v_1v_2,\ldots,v_{p-1}v_p,\allowbreak v_pu_2\}$ and $D=\emptyset$. Then $G+A-D$ is connected, so
$(A,D)$ is a solution by Lemma~\ref{lem:struct-undir}. This solution has size
$|A|+|D|=p+\frac{1}{2}|T|$, so it is optimal.
Assume that $p=0$, so $G+M$ contains only one component. If~$u_1u_2$ is
not a bridge in $G+M$, let $A=M$ and $D=\{u_1u_2\}$. Then $G+M$ is connected,
so $(A,D)$ is a solution. This solution has size $|A|+|D|=p+\frac{1}{2}|T|$, so
it is optimal.

Now assume that $u_1u_2$ is a bridge in $Q=G+M$. Let $Q_1$ and $Q_2$ denote
the components of $Q-\{u_1u_2\}$ with $u_1 \in V(Q_1)$ and $u_2 \in V(Q_2)$. Note
that $u_1u_2$ is also a bridge in~$G$.  We claim that the edges of~$M$ are
either all in $Q_1$ or all in~$Q_2$. Suppose for contradiction that $y_1z_1 \in
E(Q_1) \cap M$ and $y_2z_2 \in E(Q_2) \cap M$. Then $M'=M \setminus\{y_1z_1,y_2z_2\} \cup
\{u_1y_2,u_2y_1,z_1z_2\}$ would be a larger matching in~$\overline{G[T]}$ than
$M$, contradicting the maximality of~$M$.  Without loss of generality, we may
therefore assume that all edges of~$M$ are in $Q_1$.

Let $M=\{x_1y_1,\ldots,x_ry_r\}$, where $r=\frac{1}{2}|T|-1$.  We claim that
$u_1$ must be adjacent in~$G$ to all vertices of $T \setminus \{u_1\}$. Suppose
for contradiction that~$u_1$ is non-adjacent in~$G$ to some vertex of $T
\setminus \{u_1\}$. Since $u_1u_2 \in E(G)$, this vertex would have to be
incident to some edge in~$M$. Without loss of generality, assume $u_1x_1 \notin
E(G)$. Then $M'=M\setminus\{x_1y_1\}\cup\{u_1x_1,u_2y_1\}$ would be a larger
matching in $\overline{G[T]}$ than $M$, contradicting the maximality of $M$.
Therefore~$u_1$ is adjacent in~$G$ to every vertex of $T \setminus \{u_1\}$. In
particular, since $p=0$, it follows that $q=1$ and $G$ is connected.

Suppose that every edge between $u_1$ and $T \setminus \{u_1\}$ is a bridge in
$G$. Then no two vertices of $T \setminus \{u_1\}$ can be adjacent, and
$G[T]=K_{1,r}$. However, then Case~2 applies, which we assumed was not the
case. Without loss of generality, we may therefore assume that $u_1x_1$ is not
a bridge in $G$. Let $A=M\setminus\{x_1y_1\}\cup\{y_1u_2\}$ and $D=\{u_1x_1\}$.
Then $G+A-D$ is connected, so $(A,D)$ is a solution. Since
$|A|+|D|=\frac{1}{2}|T|-1-1+1+1=p+\frac{1}{2}|T|$, this solution is optimal.
This concludes Case~3b and therefore also concludes Case~3.

\medskip
\noindent
It is clear that $\opts(G,\delta)$ can be computed in $O(n+m)$ time.  We also
observe that the above proof is constructive, that is, we not only solve the
decision variant of \cdpe($\ea,\ed$) but we can also find an optimal solution.
To do so, we must find a maximum matching in~$\overline{G[T]}$. This takes
$O(n^{5/2})$ time by Lemma~\ref{l-mic}.
However, the bottleneck is in Case 3a, where we are glueing components by replacing two matching edges by two other matching edges, which takes $O(n^2)$ time.
As the total number of times we may need to do this is $O(n)$, this procedure may take $O(n^3)$ time in total.
Hence, we can obtain an optimal solution in
$O(n^3)$ time.
\end{proof}

\subsection{The \W-Hard Cases}\label{s-wun}

We first describe the problem used in our \W-hardness construction.
A {\em red/blue graph} is a bipartite graph $G=({\cal R},{\cal B},E)$ whose
vertices are partitioned into independent sets ${\cal R}$ (the red vertices)
and ${\cal B}$ (the blue vertices). A non-empty set $R\subseteq {\cal R}$ is an
{\em odd set} if every vertex in ${\cal B}$ has an odd number of neighbours in
$R$. The {\sc Odd Set} problem takes as input a red/blue graph $G=({\cal
R},{\cal B},E)$ and an integer $k>0$, and asks whether there is an odd set
$R\subseteq {\cal R}$ of size at most $k$. This problem is known to be
\NP-complete as well as \W-hard when parameterized
by~$k$~\cite{DowneyFVW99}. For our purposes, we need to show that the same
holds for the following restricted version of the problem.

\begin{center}
\begin{boxedminipage}{.99\textwidth}
\begin{tabular}{rl}
& \textsc{Odd-Sized Odd Set}\\
\textit{~~~~Instance:} & A red/blue graph $G=({\cal R},{\cal B},E)$ where $|{\cal R}|$ is odd, and\\
                       & an odd integer $k>0$.\\
\textit{Question:} & Is there an odd set $R\subseteq {\cal R}$ such that $|R|\leq k$ and $|R|$\\ & is odd?\\
\end{tabular}
\end{boxedminipage}
\end{center}

\begin{lemma}
\label{l-oddsized}
{\sc Odd-Sized Odd Set} is \NP-complete as well as \W-hard when
parameterized by $k$.
\end{lemma}

\begin{proof}
The {\sc Odd-Sized Odd Set} problem trivially belongs to \NP. To prove
that the problem is \NP-hard and \W-hard when parameterized
by~$k$, we give a parameterized reduction from {\sc Odd Set}. 
Recall that this problem is \NP-complete as well as \W-hard when
parameterized by~$k$~\cite{DowneyFVW99}.

Given an instance $(G,k)$ of {\sc Odd Set}, where $G=({\cal R},{\cal B},E)$ is
a red/blue graph with ${\cal R}=\{r_1,\ldots,r_p\}$ and ${\cal
B}=\{b_1,\ldots,b_q\}$ and $k$ is 
a positive integer, we construct an instance $(G',k')$ of {\sc Odd-Sized Odd
Set} as follows. We start with the disjoint union $G_1\uplus G_2$ of two copies
of $G$, where $G_i=({\cal R}_i,{\cal B}_i,E_i)$. We then add an independent set
${\cal X}=\{x_1,\ldots,x_p\}$. For each $i\in \{1,\ldots,p\}$, we make~$x_i$
adjacent to the two copies of $r_i$ in ${\cal R}_1\cup {\cal R}_2$. We then add
a vertex $r^*$ that is made adjacent to all vertices in ${\cal X}$, as well as
a vertex $b^*$ that is made adjacent to $r^*$ only. Let $G'=({\cal R}',{\cal
B}',E')$ denote the obtained red/blue graph, where ${\cal R}'={\cal R}_1\cup
{\cal R}_2\cup \{r^*\}$ and ${\cal B}'={\cal B}_1\cup {\cal B}_2 \cup {\cal
X}\cup \{b^*\}$. 
Notice that $|{\cal R'}|=2|{\cal R}|+1$ and $|{\cal R'}|$ is odd.  We set
$k'=2k+1$. 
Clearly,  $k'$ is odd.  We claim that $(G',k')$ is a yes-instance of {\sc
Odd-Sized Odd Set} if and only if $(G,k)$ is a yes-instance of {\sc Odd Set}.

First suppose $(G,k)$ is a yes-instance of {\sc Odd Set}. Then there is an odd
set $R\subseteq {\cal R}$ such that $|R|\leq k$. Consider the set $R'\subseteq
V(G')$ consisting of the two copies of $R$ in $G'$, plus the vertex $r^*$. For
each vertex $b\in {\cal B}_1\cup {\cal B}_2$, the number of vertices $b$ has in
$R'$ equals the number of neighbours the corresponding vertex in ${\cal B}$ has
in $R$. Since $R$ is an odd set in $G$, this number is odd for every vertex in
${\cal B}_1\cup {\cal B}_2$. Let $x_i\in {\cal X}$. If $r_i\in R$, then $x_i$
has three neighbours in $R'$, namely the two copies of $r_i$ in ${\cal R}_1\cup
{\cal R}_2$ and vertex $r^*$. If $r_i\notin R$, then $r^*$ is the only
neighbour of $x_i$ in $R'$. Finally, $b^*$ has exactly one neighbour in $R'$,
namely $r^*$. This proves that $R'$ is an odd set. Since $|R'|=2|R|+1\leq
2k+1=k'$ and $|R'|$ is odd, we conclude that $(G',k')$ is a yes-instance of
{\sc Odd-Sized Odd Set}. 

Now suppose that $(G',k')$ is a yes-instance of {\sc Odd-Sized Odd Set}, and
let $R'\subseteq {\cal R}'$ be an odd set in~$G'$ such that $|R'|\leq k'$ and
$|R'|$ is odd. Since $r^*$ is the only neighbour of $b^*$ in~$G'$, it holds
that $r^*\in R'$. This implies that every vertex in ${\cal X}$ must have either
two or zero neighbours in $R'\setminus \{r^*\}$. Let ${\cal X}'$ be the set
consisting of those vertices in ${\cal X}$ that have exactly two neighbours in
$R'\setminus \{r^*\}$. Since no two vertices in ${\cal X}$ have a common
neighbour other than~$r^*$ and $|R'|\leq k'=2k+1$, we find that $|{\cal
X}'|\leq k$. Let $R=\{r_i \in {\cal R} \mid x_i\in {\cal X}'\}$, and let~$R_1'$
and $R_2'$ denote the corresponding vertices in~${\cal R}_1$ and ${\cal R}_2$,
respectively. For each $x_i\in {\cal X}'$, the two neighbours of $x_i$ other
than $r^*$ are exactly the two copies of $r_i$ in $G'$. This implies that
$|R_1'|=|R_2'|=|{\cal X}'|\leq k$. By the definition of $R_1'$ and the
construction of $G'$, every vertex in ${\cal B}_1$ has an odd number of
neighbours in $R_1'$. Consequently, every vertex in ${\cal B}$ has an odd
number of neighbours in $R$. This implies that $R$ is an odd set in $G$ of size
at most $k$.
\end{proof}
We are now ready to prove the hardness results of this section.

\begin{sloppypar}
\begin{theorem}\label{thm:vertex-undir}
Let $\{\vd\} \subseteq  S\subseteq \{\vd,\ed,\ea\}$. Then \cdpe$(S)$ is
\NP-complete and \W-hard when parameterized by $k$, even if
$\delta\equiv 0$.
\end{theorem}
\end{sloppypar}

\begin{proof}
The \cdpe($S$) problem clearly belongs to \NP. To prove that the problem
is \NP-complete and \W-hard when parameterized by $k$, even if
$\delta\equiv 0$, we reduce from {\sc Odd-Sized Odd Set}. The latter problem is
\NP-complete as well as \W-hard when parameterized by $k$ due to
Lemma~\ref{l-oddsized}, and this clearly remains true when we assume that
$|{\cal R}|\geq 2$ and every vertex in ${\cal R}$ has at least one neighbour in
${\cal B}$.

\begin{sloppypar}
Let $(G,k)$ be an instance of {\sc Odd-Sized Odd Set}, where $G=({\cal R},{\cal
B},E)$ is a red/blue graph with ${\cal R}=\{r_1,\ldots,r_p\}$ and ${\cal
B}=\{b_1,\ldots,b_q\}$, and where $|{\cal R}|\geq 2$ and every vertex in ${\cal
R}$ has at least one neighbour in ${\cal B}$. We construct a graph $G^*$ as
follows. We start with two copies ${\cal B}_1,{\cal B}_2$ of ${\cal B}$, as
well as $k$ copies ${\cal R}_1,\ldots,{\cal R}_k$ of ${\cal R}$. Let ${\cal
B}^*={\cal B}_1\cup {\cal B}_2$ and ${\cal R}^*=\bigcup_{i=1}^k {\cal R}_i$.
For any two vertices $u\in {\cal B}^*$ and $v\in {\cal R}^*$, we add the edge
$uv$ if and only if the corresponding vertices in $G$ are adjacent. For every
vertex $b\in {\cal B}$, we add an edge between $b'$ and $b''$ in~$G^*$ if and
only if $b$ has even degree in~$G$, where $b',b''$ denote the copies of $b$ in
${\cal B}_1$ and ${\cal B}_2$, respectively. For every $i\in \{1,\ldots,k\}$,
we add an independent set ${\cal X}_i$ 
of size $2(k+1)$, and make all the vertices in ${\cal X}_i$ adjacent to every
vertex in~${\cal R}_i$. Let ${\cal X}^*=\bigcup_{i=1}^k {\cal X}_i$. Finally,
we add two vertices $y_1,y_2$ and make each of them adjacent to every vertex in
${\cal B}^*$. This completes the construction of $G^*$. We define a parity
function $\delta: V(G^*)\rightarrow \{0,1\}$ by setting $\delta(v)=0$ for every
$v\in V(G^*)$.
\end{sloppypar}

We will show that $(G^*,k,\delta)$ is a yes-instance of \cdpe($S$) if and only
if $(G,k)$ is a yes-instance of {\sc Odd-Sized Odd Set}. We first make some
observations about the vertex degrees in $G^*$. Recall that both $|{\cal R}|$
and $k$ are odd by the definition of {\sc Odd-Sized Odd Set}. With this in
mind, it is easy to verify that every vertex in ${\cal B}^*\cup {\cal X}^*$ has
odd degree, while every vertex in ${\cal R}^*\cup \{y_1,y_2\}$ has even degree.

Suppose $(G,k)$ is a yes-instance of {\sc Odd-Sized Odd Set}. Then there exists
an odd set $R\subseteq {\cal R}$ in $G$ such that $|R|\leq k$ and $|R|$ is odd.
Fix an arbitrary order on the vertices of $R$. For each $i\in
\{1,\ldots,|R|\}$, delete from ${\cal R}_i$ the copy of the $i$th vertex of
$R$. If $|R|<k$, then for each $i\in \{|R|+1,\ldots,k\}$, we delete the copy of
$r_1$ from ${\cal R}_i$ (regardless of whether or not $r_1\in R$); since~$|R|$
is odd and $k$ is odd, we delete an even number of copies of $r_1$
in this second step. Let $G'$ denote the obtained graph. Observe that we
obtained $G'$ from $G^*$ by deleting exactly $k$ vertices. We claim that~$G'$
is Eulerian.

Since we deleted exactly one vertex from each set ${\cal R}_i$, the degree of
each vertex in ${\cal X}^*$ decreased by exactly~$1$, making the degrees of all
these vertices even. Consider an arbitrary vertex $b\in {\cal B}^*$. Recall
that $b$ has odd degree in $G^*$. The vertex in $G$ corresponding to $b$ has an
odd number of neighbours in $R$ due to the fact that $R$ is an odd set. Exactly
one copy of each of these neighbours was deleted from $G^*$, plus an additional
even number of copies of $r_1$ in case $|R|<k$. This means that out of all the
neighbours of $b$ in $G^*$, an odd number are deleted, implying that $b$ has
even degree in $G'$. Now consider the degrees of the vertices in ${\cal
R}^*\cup \{y_1,y_2\}$. Observe that these vertices form an independent set in
$G^*$, and every vertex that is deleted from $G^*$ belongs to this set. Hence,
the parity of the degrees of the vertices in ${\cal R}^*\cup \{y_1,y_2\}$ does
not change, so all these vertices have even degree in $G'$. It remains to argue
that $G'$ is connected. Recall that we assume that $|{\cal R}|\geq 2$ and every
vertex in ${\cal R}$ has at least one neighbour in ${\cal B}$. Since we deleted
exactly one vertex from each set ${\cal R}_i$, there is at least one edge in
$G'$ between a remaining vertex of ${\cal R}_i$ and a vertex in ${\cal B}^*$.
This, together with the fact that the vertices in ${\cal X}^*\cup \{y_1,y_2\}$
are all present in $G'$, implies that $G'$ is connected. We conclude that $G'$
is Eulerian.

\begin{sloppypar}
For the reverse direction, suppose $(G^*,k,\delta)$ is a yes-instance of
\cdpe($S$). Then there is a sequence $L$ of at most $k$ operations from~$S$
transforming $G^*$ into a Eulerian graph $G'$. We claim that $L$ consists of
exactly $k$ vertex deletions, and that $L$ deletes exactly one vertex from each
set ${\cal R}_i$. Recall that each vertex in~${\cal X}^*$ has odd degree in
$G^*$. 
Let $i\in \{1,\ldots,k\}$. In order to change the (parity of the) degree of a
vertex $x\in {\cal X}_i$, we need to perform (at least) one of the following operations:
\end{sloppypar}
\begin{enumerate}[(i)]
\renewcommand{\theenumi}{(\roman{enumi})}
\renewcommand\labelenumi{\theenumi}
\item delete $x$, 
\item delete an edge incident with $x$, 
\item add an edge incident with $x$, or 
\item delete one of the neighbours of $x$. 
\end{enumerate}
Operations~(i)--(iii) 
leave the parity of at least two vertices in ${\cal X}_i$ unaltered.
Hence, from the construction of $G^*$ and the fact that $|L|=k$, it follows
that $L$ deletes exactly one vertex from each set ${\cal R}_i$.

Let $R^*\subseteq {\cal R}^*$ denote the set of vertices that are deleted from
$G^*$ by performing the operations in $L$. Note that $|R^*|=|L|=k$, and hence
$R^*$ has odd size. Let $R\subseteq {\cal R}$ be the set of those vertices in $G$ of
which $R^*$ contains an odd number of copies, i.e. $R=\{r_i\in {\cal R} \mid
R^* \mbox{ contains an odd number of copies of } r_i\}$. We claim that $R$ is a
solution for the instance $(G,k)$ of {\sc Odd-Sized Odd Set}. Since $|R^*|$ is
odd,~$|R|$ must be odd as well. It therefore remains to show that~$R$ is an odd
set in~$G$. For contradiction, suppose there is a vertex $b_j\in {\cal B}$ that
has an even number of neighbours in $R$. Consider the copy of $b_j$ in ${\cal
B}_1$; let us denote this copy by $b$. Recall that for every $r_i\in {\cal R}$,
vertex $b$ is adjacent either to all copies of $r_i$ in $G^*$ or to none of
these copies. The fact that $b_j$ has an even number of neighbours in $R$ implies
that $b$ is adjacent to an even number of vertices in $R^*$. This means that
the degree of $b$ in $G^*$ has the same parity as the degree of $b$ in~$G'$.
Since $b$ has odd degree in $G^*$ and $G'$ is Eulerian, we have thus obtained
the desired contradiction.
\end{proof}

\section{Connected Degree Balance Editing}\label{sec:directed}

Let $S\subseteq \{\ea,\ed,\vd\}$.  In Section~\ref{s-polyund} 
we will show that
\cdbe($S$) is polynomial-time solvable if $\{\ea\}\subseteq S \subseteq
\{\ea,\ed\}$ and
in Section~\ref{s-wund} we will show 
that it is \NP-complete
and \W-hard with parameter~$k$ if $\vd\in S$.

\subsection{The Polynomial-Time Solvable Cases}\label{s-polyund} 

Let  $\{\ea\}\subseteq S \subseteq \{\ea,\ed\}$.  
Let
$(G,\delta,k)$ be an instance of \cdbe($S$) with $G=(V,E)$. Let~$A$ be a set of
arcs not in $G$, and let~$D$ be a set of arcs in~$G$, with $D=\emptyset$ if
$S=\{\ea\}$.  We say that $(A,D)$ is a {\em solution} for $(G,\delta,k)$ if its
{\em size} $|A|+|D|\leq k$, the equation $d_H^+(u)-d_H^-(u)=\delta(u)$ holds for every vertex~$u$ and
the graph $H=G+A-D$ is connected; if~$H$ is not connected
then $(A,D)$ is a {\em semi-solution} for $(G,\delta,k)$.  Just as in
Section~\ref{s-polyun} 
we consider the optimization version for \cdbe$(S)$ and we use the same terminology.

Let $(G,\delta)$ be an instance of (the optimization version) of \cdbe$(S)$
where $G=(V,E)$.  Let $T=T_{(G,\delta)}$ be the set of vertices $v$ such that
$d^+_G(v)-d^-_G(v)\neq \delta(v)$.  Define a function $f_{(G,\delta)}:T\to
\mathbb{Z}$ by $f(v)=f_{(G,\delta)}(v)=\delta(v)-d^+_G(v)+d^-_G(v)$ for every
$v\in T$. 

We construct a directed multigraph $G_S$ with vertex set~$V$ 
and arc set
determined as follows.
If $\{\ea\}\subseteq S \subseteq \{\ea,\ed\}$, for each pair of distinct
vertices~$u$ and~$v$ in~$G$, if $(u,v) \notin E$, add the arc $(u,v)$ to~$G_S$
(these arcs are precisely those that can be added to~$G$). If $S=\{\ea,\ed\}$,
for each pair of distinct vertices~$u$ and~$v$, if $(u,v) \in E$, add the
arc $(v,u)$ to~$G_S$ (these arcs are precisely those whose reverse can be
deleted from~$G$). Note that adding a
(missing)
arc has the same effect on the degree
balance of the vertices in a digraph as deleting the reverse of the arc 
(if it exists).
Also 
observe that $G_S$ becomes a directed multigraph rather than a digraph only if
$S=\{\ea,\ed\}$ and there are distinct vertices~$u$ and~$v$ such that $(u,v)
\in E$ and $(v,u) \notin E$ applies.  Moreover,~$G_S$ contains at most two
copies of any arc, and if there are two copies of $(u,v)$ then $(v,u)$ is not
in $G_S$.

Let $F$ be a minimum directed
$f$-join in $G_S$ (if one exists).  Note that~$F$ may contains two copies of
the same arc if $G_S$ is a directed multigraph.  
Also note that for any pair of vertices
$u,v$, either $(u,v)\notin F$ or $(v,u)\notin\nobreak F$, otherwise $F'=F \setminus
\{(u,v),(v,u)\}$ would be a smaller $f$-join in $G_S$, contradicting the
minimality of $F$. 

We define two sets $A_F$ and~$D_F$ which, as we will show, correspond to a
semi-solution $(A_F,D_F)$ of $(G,\delta)$.  Initially set $A_F=D_F=\emptyset$.
Consider the arcs in $F$.   
If $F$ contains $(u,v)$ exactly once then add $(u,v)$ to  
$A_F$ if $(u,v)\notin E$ and add $(v,u)$ to $D_F$ if $(u,v)\in E$ (in this case $(v,u)\in E$ holds).
If $F$ contains two copies of $(u,v)$ then add $(u,v)$ to $A_F$ and $(v,u)$ to $D_F$;
note that by definition of~$F$ and~$G_S$, in this case $S=\{\ea,\ed\}, (u,v) \notin
E$ and $(v,u) \in E$.
Observe that the sets $A_F$ and $D_F$ are not multisets.
We need the following lemma, which consists of seven 
easy observations.

\begin{lemma}\label{l-weakly}
\begin{sloppypar}
Let $\{\ea\}\subseteq S \subseteq \{\ea,\ed\}$.  Let $(G,\delta)$ be an
instance of \cdbe$(S)$ where $G=(V,E)$.  Let $F$ be a minimum directed
$f$-join. The following statements hold.
\end{sloppypar}
\begin{enumerate}[(i)]
\renewcommand{\theenumi}{(\roman{enumi})}
\renewcommand\labelenumi{\theenumi}
\item \label{stat:uv-A-uv-not-E} If $(u,v)\in A_F$ then $(u,v)\notin E$.
\item \label{stat:uv-D-uv-E} If $(u,v)\in D_F$ then $(u,v)\in E$.
\item \label{stat:A-D-disjoint} $A_F\cap D_F=\emptyset$ and moreover, $(u,v)\in F$ if and only if $(u,v)\in A_F$ or $(v,u)\in D_F$. 
\item \label{stat:two-copies-F} There are two copies of $(u,v)$ in $F$ if and only if $(u,v)\in
A_F$ and $(v,u)\in\nobreak  D_F$.
\item \label{stat:no-ed-D-empty} If $S=\{\ea\}$, then $D_F=\emptyset$.
\item \label{stat:never-disconnect} If vertices $u$ and $v$ are joined by an arc in $G$ then they are
joined by an arc in $G+A_F-D_F$.
\item \label{stat:never-disconnect-changed} If $(u,v) \in F$ then $u$ and $v$ are connected by an arc in $G+A_F-D_F$.
\end{enumerate}
\end{lemma}

\begin{proof}
Statements~\ref{stat:uv-A-uv-not-E} and~\ref{stat:uv-D-uv-E} follow directly from the definitions of $A_F$ and
$D_F$, respectively.  The fact that $A_F\cap D_F=\emptyset$ follows directly
from Statements~\ref{stat:uv-A-uv-not-E} and~\ref{stat:uv-D-uv-E}. The second part of Statement~\ref{stat:A-D-disjoint} follows
directly from the definitions of $A_F$ and $D_F$. Statement~\ref{stat:two-copies-F} follows
directly from the definition of $A_F$ and~$D_F$.

To prove Statement~\ref{stat:no-ed-D-empty}, suppose for contradiction that $S=\{\ea\}$ and
$(u,v)\in D_F$. By Statement~\ref{stat:uv-D-uv-E}, $(u,v) \in E$.  
Since $S=\{\ea\}$, $F$ can contain at most one copy of $(v,u)$. By definition of $A_F$ and $D_F$, it follows that $(v,u) \in F$ and $(v,u) \in E$.
However, since $(u,v),(v,u) \in E$ and $S=\{\ea\}$, $(v,u)$ is not an arc in
$G_S$ by definition of $G_S$. Therefore $F$ cannot be an $f$-join in $G_S$,
which is a contradiction.

Next we consider Statement~\ref{stat:never-disconnect}. First suppose that $(u,v),(v,u) \in E$. If
$u$ and~$v$ are not connected by an arc in $G+A_F-D_F$, then $(u,v),(v,u) \in
D_F$. Then, by Statement~\ref{stat:A-D-disjoint}, $(v,u),(u,v) \in F$. However, as stated
earlier, this cannot happen, since $F$ is minimum. Now suppose $(u,v)\in E$ and
$(v,u)\notin E$. If $u$ and $v$ are not connected by an arc in $G+A_F-D_F$,
then $(u,v) \in D_F$. By Statement~\ref{stat:A-D-disjoint}, $(v,u) \in F$. Then $F$ must contain two copies of $(v,u)$, since $(v,u)\notin E$, so $(v,u) \in A_F$.
However in this case $u$ and $v$ are
connected by an arc in $G+A_F-D_F$. This completes the proof of Statement~\ref{stat:never-disconnect}.

Finally, we consider Statement~\ref{stat:never-disconnect-changed}. Suppose $(u,v) \in F$. If $(u,v) \in A_F$
then by Statement~\ref{stat:A-D-disjoint}, $(u,v)$ is an arc in $G+A_F-D_F$. Otherwise, by
Statement~\ref{stat:A-D-disjoint}, $(v,u) \in D_F$, so $(v,u)\in E$ by Statement~\ref{stat:uv-D-uv-E}. However,
in this case Statement~\ref{stat:never-disconnect} implies that $u$ and $v$ are connected by an arc in
$G+A_F-D_F$.
\end{proof}

If $X$ and $Y$ are sets, then $X \uplus Y$  is the multiset that consists of one
copy of each element that occurs in exactly one of $X$ and $Y$ and two copies
of each element that occurs in both.

The next lemma provides the starting point for our algorithm.

\begin{lemma}\label{l-onetoone}
\begin{sloppypar}
Let  $\{\ea\}\subseteq S \subseteq \{\ea,\ed\}$.  Let $(G,\delta)$ be an
instance of \cdbe$(S)$ where $G=(V,E)$. The following holds:
\end{sloppypar}
\begin{enumerate}[(i)]
\renewcommand{\theenumi}{(\roman{enumi})}
\renewcommand\labelenumi{\theenumi}
\item  \label{stat:f-join-is-semisolution} If $F$ is a minimum directed $f$-join in $G_S$, then $(A_F,D_F)$ is
a semi-solution for $(G,\delta)$ of size~$|F|$.
\item \label{stat:semisolution-is-f-join} If $(A,D)$ is a semi-solution for $(G,\delta)$, then $A\uplus D^R$
is a directed $f$-join in~$G_S$ of size $|A|+|D|$.
\end{enumerate}
\end{lemma}

\begin{proof}
First consider Statement~\ref{stat:f-join-is-semisolution}.  Suppose $F$ is a minimum directed $f$-join
in~$G_S$.  By Lemma~\ref{l-weakly}~\ref{stat:A-D-disjoint} and~\ref{stat:two-copies-F}, $(A_F,D_F)$ has size
$|A_F|+|D_F|=|F|$.

Let $H=G+A_F-D_F$.  Let $u\in V$.  Let $A^+(u)$ and $A^-(u)$ be the sets of
arcs in $F$ with~$u$ as tail or head, respectively, that were put into~$A_F$.
Let $D^+(u)$ and $D^-(u)$ be the set of arcs in $F$ with~$u$ as tail or head,
respectively, whose reverse was put into $D_F$.

Suppose $u \in V$. Define $d^+_{G_S(F)}(u)=d^-_{G_S(F)}(u)=0$ if $u$ is not in $G(F)$ and $f(u)=0$ if $u\notin T$.
Then by the
definition of a directed $f$-join, we have
\begin{align*}
\delta(u)-(d^+_G(u)-d^-_G(u)) &= f(u)\\[5pt]
&= d^+_{G_S(F)}(u)-d^-_{G_S(F)}(u)\\[5pt]
&= |A^+(u)|+|D^+(u)|-|A^-(u)| - |D^-(u)|.
\end{align*}
If $(u,v)\in A_F$ then $(u,v)\notin E$ by Lemma~\ref{l-weakly}~\ref{stat:uv-A-uv-not-E}.  If
$(u,v)\in D_F$ then $(u,v)\in\nobreak E$ by Lemma~\ref{l-weakly}~\ref{stat:uv-D-uv-E}. Moreover, in
that case, $(v,u)\in F$.  Consequently, we find that
\begin{gather*}
d^+_H(u)-d^-_H(u)\hspace*{22em}\\[5pt]
\begin{align*}
\hspace*{3em}&= d^+_G(u)-d^-_G(u)+|A^+(u)| - |A^-(u)|+ |D^+(u)| - |D^-(u)|\\[5pt]
&= d^+_G(u)-d^-_G(u)+\delta(u)-(d^+_G(u)-d^-_G(u))\\[5pt]
&= \delta(u).
\end{align*}
\end{gather*}
We conclude that $(A_F,D_F)$ is
a semi-solution for $(G,\delta)$. 

\medskip
\noindent
Now consider Statement~\ref{stat:semisolution-is-f-join}. Suppose $(A,D)$ is a semi-solution for $(G,\delta)$.  Let
$A^+(u)$ and $A^-(u)$ be the sets of arcs in $A$ with $u$ as tail or head,
respectively.  Let $D^+(u)$ and $D^-(u)$ be the set of arcs in $D$ with $u$ as
tail or head, respectively.  Let $H=G+A-D$. Let $u\in T$ (recall that $T$
consists of every vertex $u$ with $d^+_G(u)-d^-_G(u)\neq \delta(u)$).  Because
$(A,D)$ is a semi-solution, we have
\begin{gather*}
d^+_G(u)-d^-_G(u)+|A^+(u)|-|A^-(u)| - (|D^+(u)|-|D^-(u)|)\\[5pt]
\begin{align*}
\hspace*{5em}&= d^+_H(u)-d^-_H(u)\\[5pt]
&= d^+_G(u)-d^-_G(u)+\delta(u)-(d^+_G(u)-d^-_G(u))\\[5pt]
&= d^+_G(u)-d^-_G(u)+f(u),
\end{align*}
\end{gather*}
where we define $f(u)=0$ if $u\notin T$.  This leads to
\[f(u)=|A^+(u)|-|A^-(u)|-(|D^+(u)|-|D^-(u)|).\] Let $F=A\uplus D^R$.  Suppose
$(u,v)$ appears once in $F$.  Let $(u,v)\in A$. Then $(u,v)\notin E$. By
definition,~$G_S$ contains $(u,v)$.  Let $(u,v)\in D^R$. Then $S=\{\ea,\ed\}$, so $(v,u)\in E$. By
definition,~$G_S$ contains $(u,v)$. Suppose $(u,v)$ appears twice in $F$.
Then $(u,v)\in A$ and $(u,v)\in D^R$. Hence, $(u,v)\notin E$ and $(v,u)\in E$,
and moreover, $S=\{\ea,\ed\}$.  Then $(u,v)$ appears twice in $G_S$.  We
conclude that $F$ is a subset of the arcs in $G_S$.  Let $D^+(u)^R$ and
$D^-(u)^R$ be the set of arcs in $D^R$ with $u$ as tail or head, respectively.
Then $|D^+(u)^R|=|D^-(u)|$ and $|D^-(u)^R|=|D^+(u)|$. We find that, for all
$u\in V$,
\begin{align*}
d^+_{G_S(F)}(u)-d^+_{G_S(F)}(u) &= |A^+(u)|-|A^-(u)|+|D^+(u)^R|-|D^-(u)^R|\\[5pt]
&= |A^+(u)|-|A^-(u)|-(|D^+(u)|-|D^-(u)|)\\[5pt]
&= f(u).
\end{align*}
Hence, $F$ is a directed $f$-join.  It follows from the corresponding
definitions that the size of $(A,D)$ is $|A|+|D|=|A|+|D^R|=|A\uplus D^R|=|F|$.
This completes the proof of Lemma~\ref{l-onetoone}.
\end{proof}

Let $(G,\delta)$ be an instance of \cdbe$(S)$.  Let~$p=p_{(G,\delta)}$ be the
number of components of~$G$ that contain no vertex of~$T$.
Let~$q=q_{(G,\delta)}$ be the number of components of~$G$ that contain at least
one vertex of~$T$.
Let $t=\nobreak t_{(G,\delta)}=~\sum_{u\in T}|f(u)|$.

We now state the following lemma;
its proof is based on Lemmas~\ref{lem:dir-t-join},~\ref{l-weakly} and~\ref{l-onetoone}.

\begin{sloppypar}
\begin{lemma}\label{l-algo}
Let  $\{\ea\}\subseteq S \subseteq \{\ea,\ed\}$.  Let $(G,\delta)$ be an
instance of \cdbe$(S)$ with $q \geq 1$.  If~$F$ is a (given) minimum directed $f$-join in
$G_S$, then $(G,\delta)$ has a solution that has size at most
$\max\{|F|,p+q-1,p+\frac{1}{2}t\}$, which can be found in
$O(nm)$ time.
\end{lemma}
\end{sloppypar}

\begin{proof}
Let $F$ be a minimum directed $f$-join in $G_S$. If $H=G+A_F-\nobreak D_F$ is connected,
then the statement of the theorem holds by Lemma~\ref{l-onetoone}.  Suppose~$H$
is not connected.  We will try to replace arcs in $F$ to obtain a different
minimum directed $f$-join~$F'$ such that $H'=G+A_{F'}-D_{F'}$ will have fewer
components. Either this will eventually cause the graph to be connected (in
which case the corresponding solution will still have size~$|F|$), or else the
structure of this directed $f$-join will enable us to find a solution for
\cdbe($S$) of size either $p+q-1$ or $p+\frac{1}{2}t$. Our changes to $F$ will
be such that no additional arcs are ever added to the corresponding set $D_F$.
Thus, if $S=\{\ea\}$, then the property $D_F=\emptyset$ will be preserved.

By Lemma~\ref{lem:dir-t-join}, $G_S(F)$ must only consist of mutually
arc-disjoint directed paths from vertices~$u$ with $f(u)>0$ to vertices $v$
with $f(v)<0$. We claim that all such paths must be of length at most~2.
Suppose, for contradiction, that there is a directed path of length at least~3
in $G_S(F)$ from some vertex~$u$ to some vertex~$v$. Note that $u$ and $v$ are
in the same component of $H$. Since $H$ is not  connected, there must be a
vertex $x$ in some other component of~$H$. By Lemma~\ref{l-weakly}~\ref{stat:never-disconnect}, this
means that $x$ is not in the same component of $G$ as~$u$ or $v$, so $(u,x)$
and $(x,v)$ are arcs in $G_S$.  Replacing the directed path from $u$ to $v$ in
$F$ by the arcs $(u,x),(x,v)$ would yield a smaller directed $f$-join in $G_S$,
which is a contradiction.  Therefore all directed paths in $G_S(F)$ must be of
length at most~2. 

Let $(u,v)$ and $(u',v')$ be arcs in $F$.  Note that by
Lemma~\ref{l-weakly}~\ref{stat:never-disconnect-changed}, $u$ and $v$ are in the same component of $H$ and
$u'$ and $v'$ are in the same component of $H$.  Suppose that $(u,v)$ and
$(u',v')$ are chosen such that $u$ and $v$ are in a different component of $H$
from the one containing $u'$  and~$v'$ and that one of the following situations
holds:
\begin{enumerate}[(i)]
\renewcommand{\theenumi}{(\roman{enumi})}
\renewcommand\labelenumi{\theenumi}
\item \label{case:uv-in-A-no-bridge} either $(u,v) \in A_F$ and $(u,v)$ is not a bridge in $H$, or
\item \label{case:vu-in-D} $(v,u) \in D_F$.
\end{enumerate}
By Lemma~\ref{l-weakly}~\ref{stat:never-disconnect}, vertex $u$ is not in the same component of $G$ as
$v'$ and vertex~$v$ is not in the same component of $G$ as $u'$. Hence, by the
definition of $G_S$, the arcs $(u,v')$ and $(u',v)$ are in $G_S$.  As such, we
may replace $(u,v)$ and $(u',v')$ in $F$ by $(u,v')$ and $(u',v)$. This yields
another minimum directed $f$-join in $G_S$ which, as we explain below, reduces
the number of components in $H$ by one.  Because~$u$ and $v$ are not in the
same components of $G$ as $u'$ or $v'$, adding $(u,v)$ and $(u',v')$  to $F$
means that these two arcs will be put into $A_F$.  Suppose~\ref{case:uv-in-A-no-bridge} holds. Then the
vertices in the original component of~$H$ that contained $u$ and $v$ will still
be connected, whereas the vertices in the original component of~$H$ that
contained~$u'$ and $v'$ will still be connected as well (if necessary via a path
that uses the new arcs $(u,v)$ and $(u',v')$). Thus, $H$ has one component
less.  Suppose~\ref{case:vu-in-D} holds.  Then removing $(v,u)$ from~$F$ means removing it
from $D_F$. Hence, in $H$, the arc $(v,u)$ is restored and we can apply the
same arguments.

We apply the above replacement operation exhaustively. At termination, we have
modified~$F$ into a minimum directed $f$-join of $G_S$, in which either every
arc in $A_F$ will be a bridge in $H$ and $D_F=\emptyset$, or the end-vertices
of every arc in $F$ will all be in the same component of~$H$.  We discuss these
two cases separately.

\medskip
\noindent
\displaycase{Case 1:} {\em Every arc in $A_F$ is a bridge in $H$ and $D_F=\emptyset$}.\\ 
Then $F=A_F$. We claim that every directed path in $G_S(F)$ has length~1.  For
contradiction, suppose $(u,v)$ and $(v,w)$ are two arcs in $F$.  Since both
$(u,v)$ and $(v,w)$ are bridges in $H$, we must have that $(u,w)$ is not an arc
in $H$.  Then replacing $(u,v)$ and $(v,w)$ in $F$ by $(u,w)$ would yield a
smaller directed $f$-join in $G_S$, which would contradict the minimality of
$F$.

As every directed path in $G_S(F)$ has length~1, every arc $(u,v) \in F$ must
be such that $f(u)>0$ and $f(v)<0$.  Hence, $F=A_F$ contains exactly $\frac{1}{2}t$
arcs.

Let $H_1,\ldots,H_k$ be the components of~$H$. Because every arc in $A_F$ is a
bridge in $H$ and $D_F=\emptyset$, we find that $k=p+q-\frac{1}{2}t$.  Suppose
$k=1$. Then~$H$ is connected, so $p=0$.  Hence we
have a solution for \cdbe$(S)$ that uses $p+\frac{1}{2}t$ arcs.  Suppose $k\geq 2$.
Choose an arc  $(u,v) \in A_F$ arbitrarily and assume without loss of
generality that $u$ and $v$ are in~$H_1$. Next, choose a vertex $v_i$ in $H_i$
for $i\in\{2,\ldots,k\}$.  Replace the arc $(u,v)$ in $A_F$ by the arcs
$(u,v_2),(v_2,v_3),\ldots,(v_{k-1},v_k),(v_k,v)$.  This gives a solution for
\cdbe$(S)$ that uses $\frac{1}{2}t-1+k=\frac{1}{2}t-1+p+q-\frac{1}{2}t=p+q-1$
arcs. 

\medskip
\noindent
\displaycase{Case 2:} {\em The end-vertices of each arc in $A_F\cup D_F$ are all in the same component~of~$H$}.\\  
Suppose $H$ has at least one other component; let $x$ be a vertex in such a
component.  Suppose that $(u,v)$ and $(v,w)$ are two distinct arcs in $F$ such
that the following situation holds:~$u$ and~$v$ are in the same component of
the graph obtained from $H$ after removing $(u,v)$ and $(v,w)$.  Because $F$ is
a minimum directed $f$-join, $u$ and $w$ are distinct vertices.  By
Lemma~\ref{l-weakly}~\ref{stat:never-disconnect}, vertices $u$ and $w$ are not in the component of $G$
that contains $x$.  Hence, by the definition of~$G_S$, the arcs $(u,x)$ and
$(x,w)$ are in $G_S$.  As such, we may replace $(u,v)$ and $(v,w)$ in~$F$ by
$(u,x)$ and $(x,w)$.  This yields another minimum directed $f$-join in $G_S$
which, as we explain below, reduces the number of components in $H$ by one.

Because $u$ and $w$ are not in the component of $G$ that contains $x$, we find
that $(u,x)$ and $(x,w)$ will be put into $A_F$. Because $F$ is a minimum
directed $f$-join, $(u,w)$ must be in $H$ already, so $(u,w) \in E$ or $(u,w)
\in F$. By Lemma~\ref{l-weakly}~\ref{stat:never-disconnect} and~\ref{stat:never-disconnect-changed},~$u$ and $w$ are still in the
same component after our replacement.  Consequently, all vertices~$u,v,w,x$
will be in the same  component.  Hence, the number of components in $H$ is
reduced by one.

We apply the above replacement operation exhaustively.  If $H$ becomes
connected, then since $F$ is (still) a minimum directed $f$-join, we have found
a solution of size $|F|$.  Assume $H$ does not become connected. Then, at
termination of our procedure, we have obtained the following situation.  For
every two distinct arcs $(u,v)$ and $(v,w)$, we have that $u$ and $v$ are in
different components of the graph~$H'$ obtained from $H$ after removing $(u,v)$
and $(v,w)$. Moreover, $w$ is in the same component of $H'$ as~$u$ (by our
earlier arguments, we have that $(u,w)\in H$).

Let $H'_v$ be the component of~$H'$ that contains~$v$.
We claim that $(u,v)\in A_F$ and $(v,w)\in A_F$, and that $H'_v$ contains no vertices incident to arcs in $F \setminus
\{(u,v),(v,w)\}$. This can be seen as follows.  Because~$H'_v$ does not contain $u$ or $w$, we find that $(u,v)$ and $(v,w)$ are both
in $A_F$ due to Lemma~\ref{l-weakly}~\ref{stat:never-disconnect-changed}. If~$H'_v$
contains a vertex incident to some arc in $F\setminus
\{(u,v),(v,w)\}$, then this component must also contain the other end-vertex of
this arc by Lemma~\ref{l-weakly}~\ref{stat:never-disconnect-changed}. Suppose $u',v'$ are in~$H'_v$ and $(u',v')\in F\setminus \{(u,v),(v,w)\}$. (Note
that we do not insist that $u'\neq v$ or $v' \neq\nobreak v$.) Then we find a smaller
directed $f$-join of $G_S$ by replacing $(u,v)$, $(v,w)$ and $(u',v')$ in~$F$
by the arcs $(u,v')$ and $(u',w)$ (which are not in $F\setminus\{(u,v),
(v,w)\}$ already due to Lemma~\ref{l-weakly}~\ref{stat:never-disconnect}). This contradicts the
minimality of $F$.

We now do as follows.  Recall that every directed path in $F$ has length at
most~2.  Hence, we can partition $F$ into $r$ arcs $(u,w)$ with $f(u)>\nobreak 0$ and
$f(w)<0$ and $\frac{1}{2}t-r$ pairs of arcs $(u,v),(v,w)$ with $f(u)>0$ and
$f(w)<\nobreak 0$.  We deduced above that every directed path $(u,v)$, $(v,w)$ reduces
the number of components in $H$ by one.  Hence, the number of components in~$H$
is $1+p-(\frac{1}{2}t-r)$.  

Let $G_1,\ldots,G_k$ be the components of~$H$ that do not contain any
vertex~$v$ with $f(v) \neq\nobreak 0$. Note that $k=p-(\frac{1}{2}t-r)$.
Because~$H$ is not connected and every vertex $v$ with
$f(v)\neq 0$ belongs to the same component of~$H$, we find that $k\geq 1$.
Choose an arbitrary arc $(u,v)$ from~$F$ and for $i\in \{1,\ldots,k\}$,  choose
an arbitrary vertex $v_i$ in $G_i$.  Remove $(u,v)$ from $H$ if $(u,v)\in A_F$
or add $(v,u)$ to $H$  otherwise (by Lemma~\ref{l-weakly}~\ref{stat:A-D-disjoint} $(v,u)\in D_F$
if $(u,v)\notin A_F$).  Add the arcs
$(u,v_1),(v_1,v_2),\ldots,(v_{k-1},v_k),(v_k,v)$ to~$A_F$.  This gives a
solution for \cdbe$(\ea)$ that uses
$r+2(\frac{1}{2}t-r)+p-(\frac{1}{2}t-r)=p+\frac{1}{2}t$ arcs.

\medskip
\noindent
It is readily seen that all steps in the algorithm described above cost $O(nm)$
time. 
This completes the proof of Lemma~\ref{l-algo}.
\end{proof}

\noindent
The next result is our first main result of this section. 
We prove it by showing that the upper bound in Lemma~\ref{l-algo} is also a lower bound for (almost) any instance of \cdbe$(S)$ with $\{\ea\}\subseteq S \subseteq \{\ea,\ed
\}$ that has a semi-solution.

\begin{theorem}\label{thm:edit-dir}
For $\{\ea\}\subseteq S \subseteq \{\ea,\ed\}$,  \cdbe$(S)$ can be
solved in time $O(n^3\log n \log \log n)$.
\end{theorem}

\begin{proof}
\begin{sloppypar}
Let  $\{\ea\}\subseteq S \subseteq \{\ea,\ed\}$, and let $(G,\delta)$ be an
instance of \cdbe$(S)$.  We first use Lemma~\ref{lem:dir-t-join} to check
\end{sloppypar}
whether $G_S$ has a directed $f$-join.  Because $G_S$ has at most $2n^2$ arcs,
this takes $O(n^3\log n \log \log n)$ time.  If $G_S$ has no directed $f$-join
then $(G,\delta)$ has no semi-solution by Lemma~\ref{l-onetoone}, and thus no
solution either.  
Assume that $G_S$ has a directed $f$-join, and let $F$ be a minimum directed
$f$-join that can be found in time $O(n^3\log n \log \log n)$ by Lemma~\ref{lem:dir-t-join}. 
As before,~$p$ denotes the
number of components of~$G$ that do not contain any vertex of~$T$, while
$q$ is the number of components of~$G$ that contain at least
one vertex of~$T$, 
and $t=\sum_{u\in T}|f(u)|$. 

We will prove the following series of statements.

\begin{itemize}
\item $\opts(G,\delta)=0$ if $p \leq 1$, $q=0$,
\item $\opts(G,\delta)=p$ if $p \geq 2$, $q=0$,
\item $\opts(G,\delta)=\max(|F|,p+q-1,p+\frac{1}{2}t)$ if $q >0$.
\end{itemize}

If $p \leq1$ and $q=0$ then $A=D=\emptyset$ is an optimal solution.  If
$p\geq 2$ and $q=0$, to ensure connectivity and preserve degree balance, for
every component of~$G$ there must be at least one arc whose head is in this
component and at least one arc whose tail is in this component, thus any
solution must contain at least~$p$ arcs. Let $G_1,\ldots,G_p$ be the components
of $G$ and arbitrarily choose vertices $v_i \in V(G_i)$ for $i\in
\{1,\ldots,p\}$. Let $A=\{(v_1,v_2),(v_2,v_3),\ldots,(v_{p-1},v_p),(v_p,v_1)\}$
and $D=\emptyset$.  Then $(A,D)$ is a solution which has size~$p$ and is
therefore optimal.

Suppose $q\geq 1$.  By Lemma~\ref{l-algo} we find a solution $(A,D)$ for
$(G,\delta)$ of size at most $\max\{|F|,\allowbreak p+\nobreak q-\nobreak 1,\allowbreak p+\nobreak \frac{1}{2}t\}$ in $O(nm)$
time.  Hence, the total running time is $O(n^3\log n \log \log n)$, and it
remains to show that any solution has size at least
$\max(|F|,p+q-1,p+\frac{1}{2}t)$. 

Let $(A,D)$ be an arbitrary solution. Then $(A,D)$ is also semi-solution. Every
semi-solution has size at least~$|F|$ by Lemma~\ref{l-onetoone}~\ref{stat:semisolution-is-f-join}. Therefore
$(A,D)$ has size at least $|F|$.

Since there are $p+q$ components in $G$, we must add at least $p+q-1$ arcs to
ensure $G+A-D$ is connected.  Therefore $(A,D)$ has size at least $p+q-1$. 

Finally, for every vertex $u$ with $f(u)>0$
(resp. $f(u)<0$) we find that $(A,D)$ must be such that at least $|f(u)|$ arcs
are either in $A$ and have $u$ as a tail (resp. head) or else are in~$D$ and
have $u$ as a head (resp. tail). For every component containing only vertices
$v$ with $f(v)=0$, there must be at least one arc in $A$ whose head is in this
component and at least one arc in $A$ whose tail is in this component (to
ensure connectivity and to ensure that the degree balance is not changed for any
vertex in this component).  Therefore we have that $(A,D)$ has size at least
$p+\frac{1}{2}t$. This completes the proof of Theorem~\ref{thm:edit-dir}.
\end{proof}

\subsection{The \W-Hard Cases}\label{s-wund}

Recall that Cygan et al.~\cite{CyganMPPS14} proved that \cdbe($\{\vd\}$) is
\NP-complete and \W-hard when parameterized by $k$, even when
$\delta\equiv 0$. Our next results shows that this remains true if we allow not
only vertex deletions, but also edge deletions and/or edge additions.

\begin{sloppypar}
\begin{theorem}\label{thm:vertex-dir}
Let $\{\vd\} \subseteq  S\subseteq \{\vd,\ed,\ea\}$. Then \cdbe$(S)$ is
\NP-complete and \W-hard when parameterized by $k$, even if
$\delta\equiv 0$.
\end{theorem}
\end{sloppypar}

\begin{proof}
Let $\{\vd\} \subseteq  S\subseteq \{\vd,\ed,\ea\}$. The \cdbe($S$) problem
trivially belongs to \NP. To prove hardness, we describe a parameterized
reduction from {\sc Directed Balanced Node Deletion}. This problem takes as
input a digraph $G$ and an integer $k>0$, and asks whether there exists a
set~$A$ of at most $k$ vertices whose deletion yields a balanced digraph. This
problem is known to be \NP-complete and \W-hard with
parameter~$k$~\cite{CyganMPPS14}.  

Let $(G,k)$ be an instance of {\sc Directed Balanced Node Deletion}, and let
$n=|V(G)|$. We construct a digraph $G'$ as follows. We start with a copy of
$G$, where for every $v\in V(G)$, we write~$v'$ to denote the copy of $v$ in
$G'$. Let $V'=\{v' \mid v\in V(G)\}$. We add $k$ isolated vertices
$v_1,\ldots,v_k$. 
For each $i\in \{1,\ldots,2k+1\}$, we construct a gadget $G_i$ consisting of
vertices $a_i,b_i,x_i^1,\ldots,x_i^{n}$ and arcs $(a_i,x_i^j)$ and
$(x_i^j,b_i)$ for every $j\in \{1,\ldots,n\}$.  We make every vertex $v\in
V'\cup \{v_1,\ldots,v_k\}$ adjacent to each of the gadgets by adding arcs
$(v,a_i)$ and $(b_i,v)$ for every $i\in \{1,\ldots,2k+1\}$. This completes the
construction of $G'$. We define a function $\delta:V(G')\rightarrow \mathbb{Z}$
by setting $\delta(v)=0$ for every $v\in V(G')$.

We claim that $(G',k,\delta)$ is a yes-instance of \cdbe($S$) if and only if
$(G,k)$ is a yes-instance of {\sc Directed Balanced Node Deletion}.

First suppose $(G,k)$ is a yes-instance of {\sc Directed Balanced Node
Deletion}. Then there is a set $A\subseteq V(G)$ of size at most $k$ such
that $G-A$ is balanced. We define a set $A'\subseteq V(G')$ of size~$k$ as
follows. If $|A|=k$, then we set $A'=\{a' \mid a \in A\}$. If $|A|<k$, then we set $A'=\{a' \mid a \in A\}\cup
\{v_1,\ldots,v_{k-|A|}\}$. We claim that $G'-A'$ is Eulerian. Since the gadgets
are connected and every vertex outside the gadgets is adjacent to each of the
gadgets, it is clear that $G'-A'$ is connected. It remains to show that every
vertex in $G'-A'$ is balanced. In~$G'$, the in- and out-degrees of each vertex
$a_i$ 
equal $n+k$ and $n$, respectively, while the in- and out-degrees of each vertex
$b_i$ equal $n$ and $n+k$, respectively.  Since each of the $k$ vertices in
$A'$ is an in-neighbour of $a_i$ and an out-neighbour of $b_i$, it holds that 
$d_{G'-A'}^+(a_i)=d_{G'-A'}^-(a_i)=d_{G'-A'}^+(b_i)=d_{G'-A'}^-(b_i)=n$ for
each $i\in \{1,\ldots,2k+1\}$.  All other vertices in the gadgets, already
balanced in $G'$, remain balanced in $G'-A'$. The same holds for the vertices
in $\{v_1,\ldots,v_k\}\setminus A'$; the in- and out-degree of each of these
vertices, both in $G'$ and in $G'-A'$, equals $2k+1$. For every vertex $v'\in
V'\setminus A'$, it holds that $d_{G'-A'}^+(v')=d_{G-A}^+(v) + 2k+1$ and
$d_{G'-A'}^-(v')=d_{G-A}^-(v) + 2k+1$. Since $d_{G-A}^+(v)=d_{G-A}^-(v)$ for
every $v\in V(G)\setminus A$ due to the assumption that $G-A$ balanced, it
holds that every $v'\in V'\setminus A'$ is balanced in $G'-A'$. We conclude
that $G'-A'$ is Eulerian.

For the reverse direction, suppose there exists a sequence $L$ of operations
from $S$ that transforms~$G'$ into a Eulerian digraph. We first argue that $L$
deletes exactly $k$ vertices from $V'\cup \{v_1,\ldots,v_k\}$. 
As we mentioned before, the in- and out-degrees of each vertex $a_i$ in $G'$
equal $n+k$ and $n$ in~$G'$, respectively, while the in- and out-degrees of
each vertex~$b_i$ in $G'$ equal $n$ and $n+k$, respectively.  Since $k>0$ by
assumption, this means that the operations in $L$ need to either delete or
balance each of the $4k+2$ vertices in the set
$Z=\{a_1,\ldots,a_{2k+1},b_1,\ldots,b_{2k+1}\}$. Since $|L|=k$ and each edge
deletion or edge addition changes the degree of at most two vertices in $Z$,
there is a gadget $G_j$ such that $L$ neither deletes a vertex of $G_j$ nor
adds or deletes an edge incident with any of the vertices of~$G_j$. The fact
that the vertices of~$G_j$, and $a_j$ and $b_j$ in particular, are balanced
after applying the operations in $L$ implies that $L$ deletes exactly $k$
in-neighbours of~$a_j$ (all of which are out-neighbours of $b_j$). We conclude
that~$L$ deletes exactly $k$ vertices from $V'\cup \{v_1,\ldots,v_k\}$.

Let $A'\subseteq V'$ be the set of at most $k$ vertices that are deleted
from~$V'$ by~$L$, and let $A=\{v\in V(G) \mid v'\in A'\}$ be the corresponding
set of vertices in $G$. Let $v\in V(G)\setminus A$. From the construction of
$G'$, it holds that 
$d_{G-A}^+(v)=d_{G'-A'}^+(v) -(2k+1)$ and
$d_{G-A}^-(v)=d_{G'-A'}^-(v')-(2k+1)$.  Since
$d_{G'-A'}^+(v')=d_{G'-A'}^-(v')$, we have that $d_{G-A}^+(v)=d_{G-A}^-(v)$.
This shows that $G-A$ is balanced, and hence $(G,k)$ is a yes-instance of {\sc
Directed Balanced Node Deletion}.
\end{proof}

\section{Conclusions}\label{sec:concl}

By extending previous work~\cite{BoeschST77,CaiY11,CyganMPPS14} we completely
classified both the classical and parameterized complexity of \cdpe($S$) and
\cdbe($S$), as summarized in Table~\ref{t-thetable}.  Our work followed the
framework used~\cite{Golovach13,MathiesonS12} for  {\sc (Connected) Degree
Constraint Editing($S$)}. 
Our study was motivated by Eulerian graphs. As such,
the variants \dpe($S$) and \dbe($S$) of \cdpe($S$) and \cdbe($S$),
respectively, in which the graph $H$ is no longer required to be connected,
were beyond the scope of this paper.  It follows from results of Cai and
Yang~\cite{CaiY11} and Cygan~\cite{CyganMPPS14}, respectively, that for
$S=\nobreak \{\vd\}$, \dpe($S$) and  \dbe($S$) are \NP-complete and, when parameterized
by $k$, \W-hard, whereas they are polynomial-time solvable for $S=\nobreak \{\ed\}$
as a result of Lemmas~\ref{lem:t-join} and~\ref{lem:dir-t-join}, respectively.
The problems \dpe$(S)$ and \dbe$(S)$ are also polynomial-time solvable if
$\{\ea\}\subseteq S \subseteq \{\ea,\ed\}$; this is in fact proven by combining
Lemmas~\ref{lem:t-join} and~\ref{lem:struct-undir} for the undirected case, and
Lemmas~\ref{lem:dir-t-join} and~\ref{l-onetoone} for the directed case.  We
expect the remaining (hardness) results of Table~\ref{t-thetable} to carry over
as well.

Let $\ell$ be an integer. Here is a natural generalization of \cdpe($S$).

\begin{center}
\begin{boxedminipage}{.99\textwidth}
\begin{tabular}{rl}
\textsc{$\ell$-CDME($S$):} & \textsc{Connected Degree Modulo-$\ell$-Editing$(S)$}\\
\textit{~~~~Instance:} & A graph $G$, integer $k$ and\\ & 
                        a function
                        $\delta\colon V(G)\rightarrow\{0,\ldots,\ell-1\}$.\\
\textit{Question:} & Can $G$ be $(S,k)$-modified into a connected graph $H$\\& with 
         $d_{H}(v)\equiv\delta(v)~(\bmod~\ell)$ for each $v\in V(H)$?
                   \end{tabular}
\end{boxedminipage}
\end{center}
Note that $2$-\textsc{CDME}($S$) is \cdpe($S$).  
The following theorem shows
that the complexity of $3$-\textsc{CDME}($S$) may differ from
$2$-\textsc{CDME}($S$).

\begin{theorem}\label{thm:mod-3}
$3$-\textsc{CDME}$(\{\ea,\ed\})$ is \NP-complete even if $\delta\equiv 2$.
\end{theorem}

\begin{proof}
Reduce from the \textsc{Hamiltonicity} problem, which is \NP-complete for
connected cubic graphs~\cite{GareyJ79}.  Let $G$ be a connected cubic graph.
Let $\delta(v)=2$ for every $v\in V(G)$, and take $k=|E(G)|-|V(G)|$.  Then $G$
has a Hamiltonian cycle if and only if $G$ can be $(S,k)$-modified into a
connected graph $H$ with $d_H(v)=2~(\bmod~3)$ for all $v\in V(H)$.
\end{proof}

\begin{sloppypar}
It is natural to ask whether
$3$-\textsc{CDME}$(\{\ea,\ed\})$ is fixed-parameter tractable with parameter~$k$.
\end{sloppypar}

Finally, another direction for future research is to investigate how the complexity of  \cdpe($S$) and \cdbe($S$) changes
if we permit other graph operations, such as edge contraction, to be in the set~$S$.
For instance, Belmonte et al.~\cite{BGHP14} considered this operation and
obtained the first results extending the work of Mathieson and Szeider~\cite{MathiesonS12} in this direction.

\bibliography{editing-euler}
\end{document}